\documentclass{fundam}

\usepackage{lineno,hyperref}

\usepackage{floatflt, wrapfig}
\usepackage{xcolor}
\usepackage{amsmath}
\usepackage{amssymb}
\usepackage{graphicx}
\graphicspath{ {images/} }
\usepackage{epstopdf}
\usepackage{enumerate}
\usepackage[shortlabels]{enumitem}
\usepackage{lineno}
\usepackage[edges]{forest}
\usepackage[linesnumbered,vlined,ruled]{algorithm2e}
\usepackage{setspace}
\usepackage[list=true]{subcaption}
\usepackage[title]{appendix}
\usepackage{multirow}
\usepackage{array}
\usepackage{tikz}
\usepackage{makecell}
\usepackage[utf8]{inputenc}
\usepackage[english]{babel}

 \newtheorem{case}{Case}
 
 \newtheorem{observation}{Observation}

\begin{document}

\setcounter{page}{1}
\publyear{22}
\papernumber{2128}
\volume{187}
\issue{1}

  \finalVersionForARXIV

\title{Gathering over Meeting Nodes in Infinite Grid\footnote{A preliminary version of this paper appeared in the 6th International Conference, CALDAM 2020, Hyderabad, India, February 13-15, 2020, }}

\author{Subhash Bhagat\\
D\'epartement d’informatique\\
Universit\'e du Qu\'ebec en Outaouais, Canada \\
subhash.bhagat.math@gmail.com
\and Abhinav Chakraborty\thanks{Address for correspondence:  Advanced Computing and Microelectronics Unit,
                   Indian Statistical Institute,  Kolkata, India. \newline \newline
          \vspace*{-6mm}{\scriptsize{Received December 2021; \ accepted August  2022.}}},  Bibhuti Das, Krishnendu Mukhopadhyaya
\\
Advanced Computing and Microelectronics Unit \\
 Indian Statistical Institute,  Kolkata, India \\
\{abhinav.chakraborty06, dasbibhuti905,  krishnendu.mukhopadhyaya\}@gmail.com}

	\maketitle

\runninghead{S.Bhagat et al.}{Gathering over Meeting Nodes in Infinite Grid}
		
		\begin{abstract}
			
			The \textit{gathering over meeting nodes} problem asks the robots to gather at one of the pre-defined \textit{meeting nodes}. The robots are deployed on the nodes of an anonymous two-dimensional infinite grid, which has a subset of nodes marked as \textit{meeting nodes}. Robots are identical, autonomous, anonymous and oblivious. They operate under an \textit{asynchronous} scheduler. They do not have any agreement on a \textit{global coordinate system}. All the initial configurations for which the problem is deterministically unsolvable have been characterized. A deterministic distributed algorithm has been proposed to solve the problem for the remaining configurations. The efficiency of the proposed algorithm is studied in terms of the number of moves required for \textit{gathering}. A lower bound concerning the total number of moves required to solve the \textit{gathering} problem has been derived.
			\end{abstract}
		

\section{Introduction}
		
 A swarm of robots consists of small and inexpensive robots that work together in a cooperative environment to execute some complex tasks. A considerable amount of research in distributed computing has been focused on robot-based computing systems because of their importance in a wide range of real-world applications, like search and rescue operations, military operations, disaster management, cleaning a big surface, etc. In robot-based computing systems, {\it gathering} is one of the most widely studied problems. The {\it gathering} problem asks the mobile robots, which are initially placed at distinct locations, to gather at a common location. The common location is not fixed a priori, and the \textit{gathering} should be finalized within a finite amount of time. The study of the \textit{gathering} problem aims to find the minimal amount of capabilities required to solve the problem. It has been extensively studied in both the continuous \cite{DBLP:conf/icalp/CieliebakFPS03,DBLP:conf/soda/AgmonP04,  DBLP:journals/jda/BhagatCM16} and discrete domains \cite{DBLP:journals/tcs/KlasingMP08,DBLP:journals/tcs/KlasingKN10,DBLP:journals/tcs/DAngeloSKN16,DBLP:journals/iandc/StefanoN17}. In the discrete environment, the robots are deployed on the nodes of an anonymous graph.
 	
 \medskip
 In this paper, we have considered the {\it gathering over meeting nodes} problem in an infinite square grid. The robots and the \textit{meeting nodes} are deployed on the nodes of the grid. We assume the robots to be:
	\begin{itemize}		
\itemsep=0.9pt
		\item \textit{Anonymous}: No unique identifiers.
		\item \textit{Autonomous}: No central controller.
		\item \textit{Identical}: Indistinguishable by their appearance.
		\item \textit{Homogeneous}: All of them execute the same deterministic algorithm.
		\item \textit{Oblivious}: No memory of past information.
		\item \textit{Silent}: No explicit means of direct communication.
		\item \textit{Disoriented}: No access to a \textit{global coordinate system}, no common compass and no agreement on chirality.
		\item \textit{Unlimited visibility}: Can perceive the entire graph.
 \end{itemize}
	
	\noindent When a robot becomes active, it operates according to \textit{Look-Compute-Move (LCM)} cycle. In the \textit{look} phase, a robot takes a snapshot of the current configuration. Based on this snapshot, it computes a destination node in the \textit{compute} phase. Note that the destination node may be its current position as well. A robot moves towards its destination node in the \textit{move} phase. If the destination node is the current position of the robot, then the robot performs a \textit{null movement}. The topology considered in this paper is an anonymous grid graph, i.e., the nodes and the edges of the input grid graph are unlabeled. In the initial configuration, the robots are placed at distinct nodes of the grid. The input graph also comprises some pre-defined fixed nodes, which are referred to as \textit{meeting nodes}. The \textit{meeting nodes} are visible to the robots during the \textit{look} phase, and they occupy distinct nodes of the grid. In the initial configuration, we assume that a robot may be deployed on a \textit{meeting node}.

\medskip	
 Based on the activation and timing of the robots, the following schedulers are considered in the literature.
	\begin{enumerate}
		\item \textit{Fully-Synchronous} (FSYNC): In the \textit{fully synchronous} (FSYNC) setting, all the robots are activated simultaneously. The activation phase of all the robots can be divided into global rounds.
		\item \textit{Semi-synchronous} (SSYNC): In this setting, a subset of robots are activated simultaneously. The activation phase of each such robot can be divided into global rounds.
		\item \textit{Asynchronous} (ASYNC): In the \textit{asynchronous} (ASYNC) setting, there is no common notion of time. The duration of each \textit{Look, Compute} and \textit{Move} phases is finite, but unpredictable. \end{enumerate}
In this paper, we have assumed that the cycles are performed \textit{asynchronously} by each robot. We assume the scheduler to be fair, i.e., each robot performs its \textit{Look-Compute-Move} (LCM) cycle within finite time and infinitely often.
	
\noindent During the \textit{Look} phase, the robots have \textit{multiplicity detection capability}. The \textit{global-strong multiplicity detection} capability of a robot allows a robot to count the actual number of robots in each node. If the robots have \textit{global-weak multiplicity detection} capability, they can only detect whether multiple robots occupy a node. They cannot count the exact number of robots composing the multiplicity. If the robots are endowed with \textit{local multiplicity detection capability} instead of the global version, then they can detect whether there is a multiplicity or the number of robots composing the multiplicity only in their current location node. In this paper, we assumed the robots to have \textit{local-weak multiplicity detection} capability, i.e., the robots can detect whether there exists a multiplicity in their current location node.

\subsection{Motivation}
 In this paper, we investigate the \textit{gathering} problem in an infinite grid, where some of the nodes are designated as \textit{meeting nodes}. In this setting, the movements of the robots are only allowed along the grid lines and the robots need to gather at one of the \textit{meeting nodes}. The fundamental motivation behind studying the {\it gathering over meeting nodes} problem in infinite grids is to investigate the solvability of the {\it gathering} problem where both the movements of the robots and the \textit{gathering} points are restricted.	
 \begin{enumerate}
     \item Gathering \cite{DBLP:journals/tcs/KlasingMP08,DBLP:journals/tcs/KlasingKN10,DBLP:journals/tcs/DAngeloSKN16,DBLP:journals/iandc/StefanoN17} problem has been studied in the discrete domain, where the movements of the robots are restricted along the edges. However, the gathering point was not restricted. The rationale behind considering the \textit{meeting nodes} might be of practical use. In general, the {\it gathering problem} requires the robots to coordinate their movements and meet at a location that they are not aware of beforehand. However, the \textit{gathering} may be limited to some specific regions or points. Another possibility is that robots may need to gather at one of the designated points in many real-life applications, e.g., one of the base stations or charging stations, etc.
     \item Cicerone et al. \cite{DBLP:journals/dc/CiceroneSN18} have studied the {\it gathering on meeting points} problem in the Euclidean plane. Though the gathering points are restricted, the robots are free to move throughout the plane. In the continuous domain, it is assumed that the robots move with high accuracy and infinite precision. In specific models, the robots can even perform \textit{guided} movements, i.e., they can move along some specified curve \cite{DBLP:journals/jpdc/PattanayakMRM19,DBLP:journals/dc/CiceroneSN19}. Moreover, the robots can move even by infinitesimally small amounts. Even if the field of robot deployment is small, a dimensionless robot can move without creating any collision. The correctness of the algorithms relies on the accurate execution of the movements. However, the vision sensors do not have infinite precision for real-life robots with weak mechanical capabilities. In the continuous domain, a robot can travel an amount of distance that may be an irrational number. In practice, it may not always be possible to perform such infinitesimal movements with infinite precision. This motivates us to consider the problem in a grid-based terrain where the movements are restricted along the grid lines, and a robot can move to one of its neighbors in one step. However, we have assumed that the movement of a robot is instantaneous, i.e., during the \textit{look} phase of each robot, the other robots are always detected on the nodes of the input grid graph. Thus, consideration of discrete environments recognizes the fact that many vision sensors produce digital and therefore, discrete snapshots of the environment. The grid topology is a natural discretization of the plane. It has numerous applications in real-life robot navigation systems, such as industrial Automated Guided Vehicles \cite{DBLP:conf/icra/BarberaQIG03} and Coverage Path Planning \cite{DBLP:conf/atal/SharmaDK19}, where grid type floor layouts can be suitably implemented. The restrictions imposed by the grid model on the movements of the robots make it difficult to design algorithms, as opposed to the movement of the robots in a continuous environment.
 \end{enumerate}

		\subsection{Earlier works}
	\textit The \textit{gathering} problem has been extensively studied in the literature \cite{ DBLP:conf/icalp/CieliebakFPS03,DBLP:conf/soda/AgmonP04,  DBLP:journals/jda/BhagatCM16,DBLP:series/lncs/Flocchini19,DBLP:series/synthesis/2012Flocchini}. In the discrete domain, the robots are deployed on the nodes of the input graph. \textit{Gathering} in the discrete domain has been largely studied in rings \cite {DBLP:journals/tcs/KlasingMP08,DBLP:journals/tcs/KlasingKN10,DBLP:journals/jda/DAngeloSN14,DBLP:journals/dc/DAngeloSN14,10.1007/978-3-642-13284-1_9,10.1007/978-3-642-32589-2_48}. Klasing et al. \cite {DBLP:journals/tcs/KlasingMP08} studied the \textit{gathering} problem in an anonymous ring and proved that \textit{gathering} in an anonymous ring is impossible without \textit{multiplicity detection capability}. With the assumption of \textit{global-weak multiplicity detection capability}, they proposed a distributed algorithm to solve the \textit{gathering} problem for all the configurations having an odd number of robots and all asymmetric configurations when the number of robots is even. Klasing et al. \cite{DBLP:journals/tcs/KlasingKN10} studied configurations in an anonymous ring which admits symmetries and having an even number of robots. They solved the problem for all configurations with more than eighteen robots. They proved that, for an odd number of robots, \textit{gathering} is feasible if and only if the configuration is not periodic. Kamei et al. \cite{10.1007/978-3-642-32589-2_48} studied the \textit{gathering} problem in anonymous rings using \textit{local weak multiplicity detection capability}. D'Angelo et al. \cite{DBLP:journals/jda/DAngeloSN14} considered the \textit{gathering} problem on anonymous rings with six robots. They proposed a distributed algorithm to solve the problem that assumes \textit{global-weak multiplicity detection} capability of the robots.
	
 Cicerone et al. \cite{DBLP:journals/tcs/CiceroneSN21}, \nocite{DBLP:conf/sirocco/CiceroneSN19,DBLP:conf/sss/CiceroneSN19,DBLP:conf/caldam/BhagatCDM20} studied the \textit{gathering} problem in complete bipartite graphs under FSYNC scheduler. They considered dense and symmetric graphs like complete graphs and complete bipartite graphs. They characterized the solvability of \textit{gathering} in such graphs. Bose et al. \cite{10.1007/978-3-030-14094-6_7} considered the \textit{gathering} problem in hypercubes. They proposed an optimal algorithm, which minimizes the total number of moves by all the robots.
	
 D'Angelo et al. \cite{DBLP:journals/tcs/DAngeloSKN16}, studied the {\it gathering} problem on trees and finite grids. They proved that even with \textit{global-strong multiplicity detection} capability, a configuration remains ungatherable if and only if it is periodic or symmetric, with the line of symmetry passing through the edges of the grid. For the remaining configurations, they solved the problem without assuming any \textit{multiplicity detection capability} of the robots. Stefano et al. \cite{DBLP:journals/dc/StefanoN17}, studied the optimal {\it gathering} of robots in anonymous graphs. They also studied the optimal \textit{gathering} problem in infinite grids \cite{DBLP:journals/iandc/StefanoN17}. In \cite{DBLP:journals/iandc/StefanoN17}, they proposed a deterministic distributed algorithm that minimizes the total distance traveled by all the robots. This paper also introduced the concept of \textit{Weber-point} on vertex-weighted graphs. A \textit{Weber-point} is a node of the graph that minimizes the sum of the length of the shortest paths from it to each robot.
	
 Fujinaga et al. \cite{10.1007/978-3-642-17653-1_1} introduced the concept of \textit{fixed points} or \textit{landmarks} on the \textit{Euclidean plane}. The \textit{landmarks covering problem} requires the robot to reach a configuration where all the robots must occupy a unique \textit{fixed} point or a \textit{landmark}. They propose an algorithm that assumes common chirality among the robots. The proposed algorithm minimizes the total distance traveled by all the robots. Cicerone et al. \cite{Cicerone2019} studied the \textit{embedded pattern formation problem} without assuming any common chirality. The problem asks for a distributed algorithm that requires the robots to occupy all the \textit{fixed points} within a finite amount of time. Each \textit{fixed point} must be occupied by exactly one robot. The \textit{$k$-circle formation} \cite{BhagatDCM21,DasCBM22} problem asks a set of robots to form disjoint circles having $k$ robots each at distinct locations. The circles are centered at the set of fixed points. Cicerone et al. \cite{DBLP:journals/dc/CiceroneSN18}, studied a variant of the {\it gathering} problem on the \textit{Euclidean plane}, where robots must gather at one of the pre-determined points, referred to as \textit{meeting points}. They defined the problem as {\it gathering on meeting points problem}. They proposed a deterministic algorithm that minimizes the total distance traveled by all the robots and minimizes the maximum distance traveled by a single robot. The proposed algorithms assume \textit{global-weak multiplicity detection capability} of the robots. In our paper, we have proposed a deterministic algorithm that assumes \textit{local-weak multiplicity detection} of the robots.

		\subsection{Our contributions}

	This paper considers \textit{gathering over meeting nodes} problem in an infinite grid by \textit{asynchronous} oblivious mobile robots. We have shown that even if the robots are endowed with \textit{multiplicity detection capability}, some configurations remain ungatherable. It includes the following collection of configurations:
	\begin{enumerate}
	    \item The configurations admitting a unique line of symmetry such that the line of symmetry does not contain any robots or {\it meeting nodes}.
	    \item The configurations admitting rotational symmetry with no robots or {\it meeting nodes} on the center of rotation.
	\end{enumerate}
We have proposed a deterministic distributed algorithm to solve the {\it gathering} problem for the remaining configurations. We have studied the efficiency of the proposed algorithm in terms of the total number of moves executed by the robots. A lower bound has been derived concerning the total number of movements performed by any algorithm for solving the \textit{gathering over meeting nodes} problem. We have proved that any algorithm that solves the \textit{gathering over meeting nodes} problem requires $\Omega(Dn)$ moves, where $D$ is the larger side of the initial minimum enclosing rectangle of all the robots and \textit{meeting nodes} and $n$ is the number of robots. Our proposed algorithm requires $O(Dn)$ moves, i.e., the algorithm is asymptotically optimal.

	\subsection{Outline}

	The following section focuses on the robot model and provides some preliminary definitions and notations. These definitions and notations are relevant in understanding the problem definition. Section \ref{s3} provides the formal description of the \textit{gathering} problem. A sufficient condition for the solvability of the \textit{gathering} task has been stated in Section \ref{s3}. In Section \ref{s4}, we have provided a deterministic distributed algorithm that solves the \textit{gathering over meeting nodes}. Section \ref{s6} provides a lower bound to the complexity of the \textit{gathering} problem in terms of the number of moves executed by the robots to finalize the \textit{gathering}. In this section, we have provided a complexity analysis for our proposed algorithm. Finally, Section \ref{s7} concludes the paper with some future directions to work with.

	\section{Model and definitions}\label{s2}

	\subsection{Model}	
	Robots are assumed to be autonomous, anonymous, homogeneous, dimensionless and oblivious. They do not have explicit means of communication. They have an unlimited and unobstructed visibility range, i.e., each robot can observe the entire grid. The robots do not have any agreement on a \textit{global coordinate system} and chirality. Each robot perceives the configuration with respect to its local coordinate system with origin as its current position. Initially, the robots are assumed to be on the distinct nodes of the input grid. Each active robot executes \textit{Look-Compute-Move(LCM)} cycle under an \textit{asynchronous scheduler}. A robot can instantly move to one of its adjacent nodes along the grid lines. The movement of a robot is instantaneous, i.e., any robot performing a \textit{Look} operation observes all the other robot's positions only at the nodes of the input grid graph.

	\subsection{Definitions}
	In this subsection, we have proposed some terminologies and definitions.
	\begin{itemize}
		\item \textbf{System Configuration:}
		\begin{itemize}
			\item $P=(\mathbb{Z}$, $E')$: infinite path graph where the vertex set corresponds to the set of integers $\mathbb {Z}$ and the edge set is denoted by the ordered pair $E'=\lbrace (i$, $i+1)| i\in \mathbb{Z}\rbrace$.
			\item \textit{Cartesian product of the graph $P\times P$}: input grid graph.
			\item $V$ and $E$: set of nodes and edges of the input grid graph, respectively.
			\item $d(u,v)$: \textit{Manhattan} distance between the nodes $u$ and $v$.
			\item $R=\lbrace r_1,r_2,\ldots,r_n\rbrace$: a set of robots deployed on the nodes of the grid.
			\item $r_i(t)$: position of the robot $r_i$ at time $t>0$. When there is no ambiguity, $r$ will represent both the robot and the position occupied by it.
			\item $R(t)=\lbrace r_1(t),r_2(t),...,r_n(t)\rbrace$: multiset of robot positions at time $t$. At $t=0$, $r_i(t) \neq r_j(t)$, for all $r_i(t), r_j(t) \in R(t)$. However, at $t> 0$, $r_i(t)$ may be equal to $r_j(t)$, for some $r_i(t), r_j(t) \in R(t)$.
			\item $M=\lbrace m_1,m_2,\ldots,m_s\rbrace$: set of {\it meeting nodes} located at the nodes of the grid graph.
			\item $C(t)=(R(t)$, $M)$: system configuration at time $t$.
		\end{itemize}
		\item\textbf{Symmetry:} An \textit{automorphism} of a graph $G=(V$, $E)$ is a bijective map $\phi:V\rightarrow V$ such that $u$ and $v$ are adjacent if and only if $\phi (u)$ and $\phi (v)$ are adjacent. \textit{Automorphism} of graphs can be extended similarly to define \textit{automorphism} of a configuration. Let $l:V\rightarrow\lbrace 0, 1, 2, 3, 4, 5 \rbrace$ be defined as a function, where:
		\[ l(v)=\begin{cases}
		0 & \text{if} \;v \textrm { is an empty node} \\
		1 & \text{if} \;v \textrm{ is a {\it meeting node}} \\
		2 & \text{if} \;v \textrm{ is a single robot position on a {\it meeting node}}\\
		3 & \text{if} \;v \textrm{ is a robot multiplicity on a \textit{meeting node}} \\
		4 & \text{if} \;v \textrm{ is a single robot position not on any \textit{meeting node}}\\
		5 & \text{if} \;v \textrm{ is a robot multiplicity not on any {\it meeting node}}
		\end{cases}
		\]
		
		An \textit{automorphism} of a configuration $C(t)$ is an \textit{automorphism} $\phi$ of the input grid graph such that $l(v)=l(\phi(v))$ for all $v\in V$. The set of all \textit{automorphisms} of a configuration forms a group which is denoted by $Aut(C(t)$, $l)$. If $\vert Aut(C(t), $ $l)\vert=1$, then the configuration is asymmetric. Otherwise, the configuration is said to be symmetric. Note that the function $l$ denotes the status or type of a node, i.e., $l(v)$ denotes whether the node $v$ is an empty node, a \textit{meeting node} without any robot positions on it, a \textit{meeting node} with a single or multiple robots on it, or a single or multiple robot positions not lying on any \textit{meeting node}. We assume that the grid is embedded in the \textit{Cartesian} plane. Hence, a grid can admit only three types of symmetry, namely, translation, reflection and rotation. Since the number of robots and {\it meeting nodes} is finite, translational symmetry is not possible. A unique line of symmetry characterizes a reflectional symmetry. The line of symmetry can be horizontal, vertical, or diagonal and can pass through the nodes or edges of the graph. The angle of rotation and the center of rotation characterize rotational symmetry. The angle of rotation can be $90^{\circ}$ or $180^{\circ}$, whereas the center of rotation can be a node, a center of an edge, or the center of a unit square.
		\item \textbf{Partitive automorphism:} Given an \textit{automorphism} $\phi$ $\in$ $Aut(C(t)$, $l)$, the cyclic subgroup of order $k$ generated by $\phi$ is given by $\lbrace\phi^{0}, \phi^{1}= \phi, \phi^2= \phi\circ\phi, \ldots, \phi^{k-1} \rbrace$, where $\phi^{0}$ denotes the identity of the cyclic subgroup. Let $H$ be any subgroup of $Aut(C(t),l)$. We define a relation $\rho$ as follows: For some $x$, $y \in V$, we say that $x$ and $y$ are related by the relation $\rho$ if and only if there exists an \textit{automorphism} $\gamma \in H$ such that $\gamma(x)=y$. Note that the relation $\rho$ is an equivalence relation defined on the set of vertices $V$. The equivalence class of the node $x$ is defined as the \textit{orbit} of $x$ \cite{DBLP:journals/iandc/StefanoN17} and is denoted by $H(x)$. These \textit{orbits} form a partition of the set $V$, since they represent disjoint equivalence classes. An \textit{automorphism} $\phi \in Aut (C(t), l)$ is called partitive on $V'' \subset V$, if the cyclic subgroup $H=\lbrace \phi^{0}, \phi^{1}= \phi, \phi^2=\phi\circ\phi,\ldots, \phi^{k-1} \rbrace$ generated by $\phi$ has order $k >1$ and is such that $|H (u)|=k$ for each $u\in V''$.
		
		Suppose a configuration admits a unique line of symmetry $L$ such that $L$ does not pass through any node. Then, there exists an \textit{automorphism} $\phi$ $\in$ $Aut$ $(C(t)$, $l)$ which is partitive on the set of nodes $V''=V$. The cyclic subgroup $H$ generated by $\phi$ with $k=2$ is given by $H=\lbrace \phi^{0}, \phi^{1}\rbrace$. Similarly, assume that a configuration admits rotational symmetry where the center of rotation $c$ is not a node. If the angle of rotation is $90^{\circ}$, then there exists an \textit{automorphism} $\phi$ $\in$ $Aut$ $(C(t)$, $l)$ which is partitive on the set of nodes $V''=V$ and the cyclic subgroup $H$ generated by $\phi$ with $k=4$ is given by $H=\lbrace \phi^{0}, \phi^{1}, \phi^{2}, \phi^{3}\rbrace$.

	\item \textbf{Configuration view}:
	Let $MER$ denote the \textit{minimum enclosing rectangle} of $R \cup M$. $MER$ is defined as the smallest grid-aligned rectangle that contains all the robots and \textit{meeting nodes}. Assume that the dimension of $MER$ is $p\times q$. We define the length of a side of $MER$ in terms of the number of grid edges on them. Let us consider the eight senary strings of length $pq$ associated with the corners of $MER$, where for each corner of $MER$, there are two senary strings defined. A senary string of length $pq$ is constructed as follows: Starting from a corner of $MER$, proceeding in the direction parallel to the width of $MER$ and scanning the entire grid sequentially, we consider all the grid lines of the $MER$ column by column. While scanning the grid, we associate $l(v)$ to each node $v$ that the string encounters. Proceeding similarly, we can define the string associated to the same corner and encounter the nodes of the grid in the direction parallel to the length of the grid. For a corner $i$, let the two strings defined are denoted by $s_{ij}$ and $s_{ik}$. Similarly, two senary strings of length $p q$ are associated with each corner of $MER$.

 \begin{figure}[!b]
				\centering
			{
				\includegraphics[width=0.339\columnwidth]{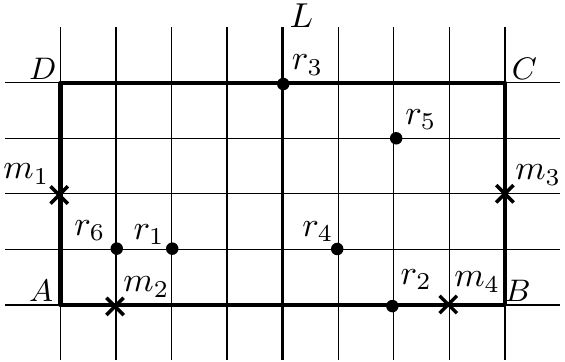}
			}
			\caption{The crosses represent \textit{meeting nodes} and the black circles represent robot positions. $L$ is the line of symmetry for the \textit{meeting nodes}. The lexicographic largest string is $s_{AD}=001001400004000000000000404000400401000000100$. $A$ is the \textit{key corner}. }
			\label{string2}
		\end{figure}
	
 First, consider the case when $MER$ is a non-square rectangle. We can distinguish the two strings associated to a particular corner by considering the string which is in the direction parallel to the side of the minimum length. Consider any particular corner $i$ of $MER$. Assume that $|ij|<|ik|$. The direction parallel to $ij$ is considered as the \textit{string direction} associated to $i$. We define $s_i=s_{ij}$ as the \textit{string representation} associated to the corner $i$. The direction parallel to the larger side (i.e., $s_{ik}$) is defined as the \textit{non-string direction} associated to the corner $i$. In the case of a square grid, between the two strings associated to a corner, the \textit{string representation} is defined as the string which is lexicographically larger, i.e., $s_i=max(s_{ij}, s_{ik})$, where the maximum is defined according to the lexicographic ordering of the strings. Note that both the strings associated with a particular corner are equal if $MER$ is symmetric with respect to the diagonal line of symmetry passing through that corner. Here, one of the directions is arbitrarily selected as the \textit{string direction}. If the configuration is asymmetric, we will always get a unique largest lexicographic string (In Figure~\ref{string2}, $s_{AD}$ is the lexicographic largest string and $AD$ is the \textit{string direction} associated to $A$). Without loss of generality, let $s_{i}$ be the largest lexicographic string among all the strings associated to the corners of $MER$. Then we refer to $i$ as the {\it key corner} (In Figure~\ref{string2}, $A$ is the \textit{key corner}). A corner which is not a \textit{key corner} is defined as a \textit{non-key corner}. The definition of the \textit{key corner} is similar to one defined by Stefano et al. \cite{DBLP:journals/iandc/StefanoN17}. The \textit{configuration view} of a node is defined as the tuple ($d', x$), where $d'$ denotes the distance of a node from the {\it key corner} in the {\it string direction} and $x$ denote the status of the node, i.e., $x=l(v)$.
 			
	\item \textbf{Symmetricity of the set $M$:} We define $MER_F$ as the smallest grid-aligned rectangle that contains all the \textit{meeting nodes}. Define the function $\lambda:V\rightarrow\lbrace 0, 1  \rbrace$ as follows:
		\[ \lambda(v)=\begin{cases}
		0 & \text{if} \;v \textrm { is not a meeting node} \\
		1 & \text{if} \;v \textrm{ is a {\it meeting node}} \\
	
		\end{cases}
		\]
We can define a string $\alpha_{i}$ similar to $s_{i}$. The only difference is that instead of $l(v)$, each node $v$ is associated with $\lambda(v)$. If the \textit{meeting nodes} are asymmetric, then there exists a unique lexicographic largest string $\alpha_i$. If the \textit{meeting nodes} are not asymmetric, then the \textit{meeting nodes} are said to be symmetric. The corner with which the lexicographic largest string $\alpha_i$ is associated is defined as the \textit{leading corner} (In Figure \ref{order1}, $\alpha_{DA}=01000 00100 00000 10100 00000 00010 01000$ is the largest lexicographic string among the $\alpha_i's$. $D$ is the \textit{leading corner}).

	\begin{figure}[!h]
\vspace*{3mm}
				\centering
			{
				\includegraphics[width=0.339\columnwidth]{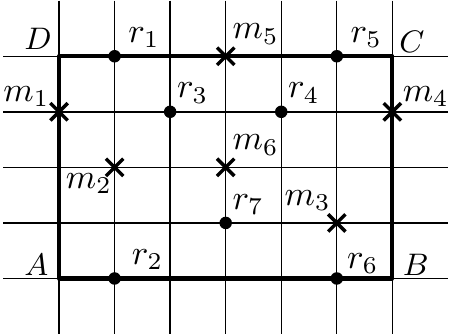}
			}
			\caption{The \textit{meeting nodes} are asymmetric. $D$ is the unique \textit{leading corner}.}
			\label{order1}
		\end{figure}\vspace*{-3mm}
		\end{itemize}
	
		\section{Gathering over meeting nodes problem} \label{s3}

		In this section, we consider the problem definition for \textit{gathering}. A distributed deterministic \textit{gathering} algorithm for \textit{gathering} $n\geq 2$ robots has been proposed. The strategy is to select a single \textit{meeting node} such that all the robots agree on it and gather at that \textit{meeting node} within a finite time. Since the robots are oblivious, the main objective is to maintain the invariance of the \textit{gathering} \textit{meeting node} during the execution of the algorithm.

	\subsection{Problem definition and impossibility results}

		In this subsection, the \textit{gathering over meeting nodes problem} is formally defined in an infinite grid.

		\subsubsection{Problem definition:}

		Given a configuration $C(t)=(R(t)$, $M)$, the {\it gathering over meeting nodes problem} in an infinite grid asks the robots to gather at one of the {\it meeting nodes} within finite time. In an initial configuration, all the robots occupy distinct nodes of the grid. We say a configuration is {\it final} at time $t$ if the following conditions hold:
		\begin{itemize}
		    \item all the robots are on a single \textit{meeting node}.
		    \item each robot is stationary.
		    \item any robot taking a snapshot in the \textit{look} phase at time $t$ will decide not to move.
		\end{itemize}
		Our proposed algorithm is a deterministic distributed algorithm that gathers all the robots at a single \textit{meeting node} within a finite amount of time.

		\subsubsection{Partitioning of the initial configuration:}

		All the configurations can be partitioned into the following disjoint classes.
		\begin{enumerate}[1.]
			\item $\mathcal{I}_1-$ This includes the following class of configurations.
			\begin{enumerate}
			    \item $\mathcal I_{11}-$ $M$ is asymmetric (Figure~\ref{order1}).
			    \item $\mathcal I_{12}-$ $M$ is symmetric with respect to a unique line of symmetry $L$ and there exists at least one \textit{meeting node} on $L$. $R \cup M$ is either asymmetric or symmetric with respect to $L$ (Figure \ref{part2}(a)).
			     \item $\mathcal I_{13}-$ $M$ is symmetric with respect to rotational symmetry with $c$ as the center of rotation and there exists a \textit{meeting node} on $c$. $R \cup M$ is either asymmetric or symmetric with respect to rotational symmetry (Figure \ref{part3}(a)).
			   \end{enumerate}
			   \item $\mathcal{I}_2-$ This includes the following class of configurations.
			
  		\begin{figure}[!h]
  \vspace*{-2mm}
			\centering
					{
				\includegraphics[width=0.270\columnwidth]{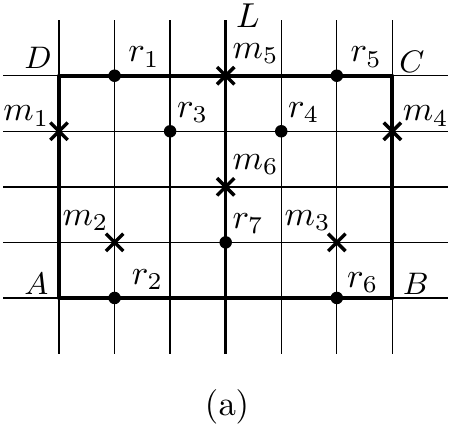}
			}
			\hspace*{0.51cm}
				{
				\includegraphics[width=0.270\columnwidth]{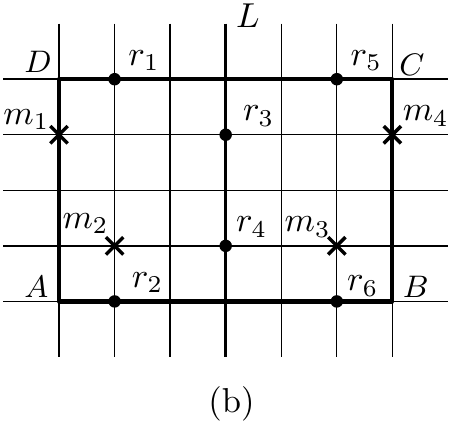}
			}
			\hspace*{0.51cm}
			{
				\includegraphics[width=0.270\columnwidth]{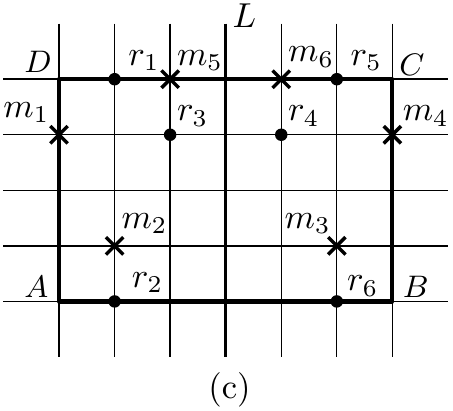}
			}\vspace*{-2mm}
			\caption{(a) $\mathcal{I}_{12}$- configuration. There exist {\it meeting nodes} on $L$. (b) $\mathcal{I}_{31}$- configuration. There exist no {\it meeting nodes} on $L$, but there exists a robot position on $L$. (c) $\mathcal{I}_{41}$- configuration without robots or {\it meeting nodes} on $L$.}
			\label{part2}
		\end{figure}
		
	\begin{figure}[!h]
			\centering
			{
				\includegraphics[width=0.25\columnwidth]{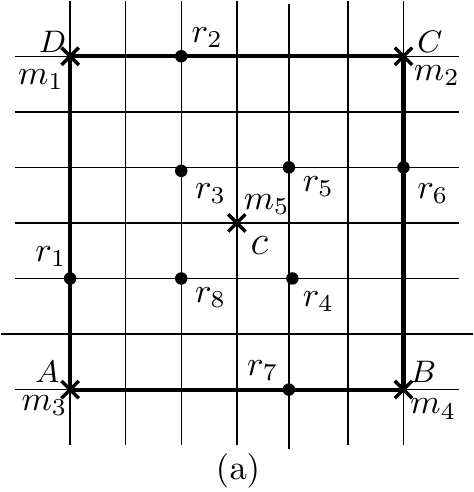}
			}
			\hspace*{0.51cm}
			{
				\includegraphics[width=0.25\columnwidth]{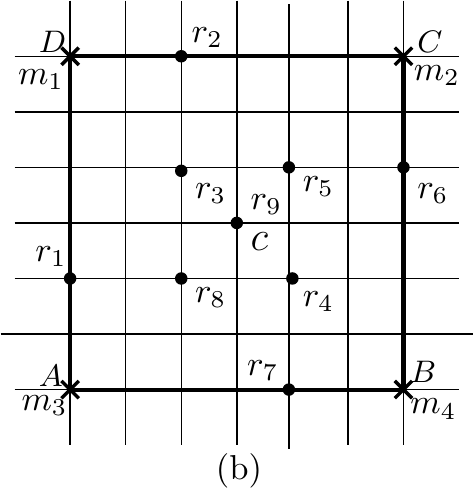}
			}
			\hspace*{0.51cm}
			{
				\includegraphics[width=0.25\columnwidth]{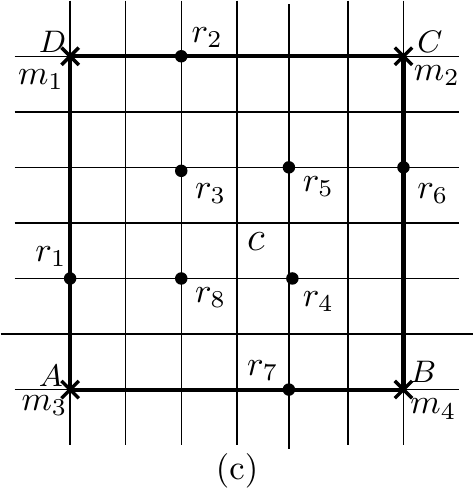}
			}\vspace*{-2mm}
			\caption{(a) $\mathcal{I}_{13}$- configuration with a \textit{meeting node} on $c$. (b) $\mathcal{I}_{32}$- configuration without a {\it meeting node} on $c$, but there exists a robot position on $c$. (c) $\mathcal{I}_{42}$- configuration without a robot or {\it meeting node} on $c$.}
			\label{part3}\vspace*{-3mm}
		\end{figure}

			   \begin{enumerate}
			       \item $\mathcal I_{21}-$ $M$ is symmetric with respect to a unique line of symmetry $L$. $R \cup M$ is asymmetric and there does not exist any \textit{meeting node} on $L$ (Figure~\ref{string2}).
			       \item $\mathcal I_{22}-$ $M$ is symmetric with respect to rotational symmetry. $R \cup M$ is asymmetric and there does not exist a \textit{meeting node} on $c$.
			   \end{enumerate}
			   \item $\mathcal{I}_3-$ This includes the following class of configurations.
			\begin{enumerate}
			    \item $\mathcal I_{31}- $ $M$ is symmetric with respect to a unique line of symmetry $L$. $R\cup M$ is symmetric with respect to $L$. There does not exist any \textit{meeting node} on $L$, but there exists at least one robot position on $L$ (Figure \ref{part2}(b)).
			    \item $\mathcal I_{32}- $ $M$ is symmetric with respect to rotational symmetry with $c$ as the center of rotation. $R\cup M$ is symmetric with respect to rotational symmetry. There does not exist a \textit{meeting node} on $c$, but there exists a robot position on $c$. $R \cup M$ may contain either no line of symmetry or at least one line of symmetry (Figure \ref{part3}(b)).
			\end{enumerate}
			\item  $\mathcal{I}_4-$ This includes the following class of partitive configurations.
			\begin{enumerate}
			    \item $\mathcal I_{41}- $ $M$ is symmetric with respect to a unique line of symmetry $L$. $R\cup M$ is symmetric with respect to $L$ and there does not exist any \textit{meeting node} or robot position on $L$ (Figure \ref{part2}(c)).
			     \item $\mathcal I_{42}- $ $M$ is symmetric with respect to rotational symmetry with $c$ as the center of rotation. $R\cup M$ is symmetric with respect to rotational symmetry and there does not exist any \textit{meeting node} or robot position on $c$. $R \cup M$ may contain either no line of symmetry or at least one line of symmetry (Figure \ref{part3}(c)).
			\end{enumerate}
			\end{enumerate}

	\noindent Let $\mathcal I$ denote the set of all initial configurations. Each time a robot is active, it observes the configuration in its \textit{Look} phase and determines the current class of configuration in which it belongs without any conflict.
	 	
	\subsection{Impossibility result}

	In this subsection, we provide a sufficient condition for solving the \textit{gathering over meeting nodes} problem in an infinite grid.
	\begin{theorem} \label{impo}
	Given an initial configuration $C(0)$, let $V' \subset V$ be a subset of nodes such that $V' \cap R (0)= \emptyset$. If there exists an automorphism $\phi$ that is partitive on the set $V \setminus V'$ and $\phi(v') =v'$, for each $v' \in V'$, then there does not exist any deterministic algorithm which can ensure gathering on a node in $V \setminus V'$.
	\end{theorem}
	\begin{proof}
	  Let $H$ be the cyclic subgroup generated by $\phi$ and $k >1$ be the size of the corresponding orbits. If possible, let algorithm $\mathcal A$ solve the {\it gathering over meeting nodes} problem and ensure gathering over a \textit{meeting node} $m \in V \setminus V'$. This implies that starting from $C(0)$, all the robots reach a \textit{final configuration}. Consider the scheduler to be \textit{fully-synchronous}. Suppose, in the initial configuration, there exists a robot $r$ on a node $v \in V \setminus V'$ in the input grid graph. Since the scheduler is assumed to be \textit{fully-synchronous}, all the robots in the \textit{orbit} $H(v)$ are activated at the same time. As each robot in $H(v)$ has identical views, $\mathcal A$ cannot distinguish the robots in $H(v)$ deterministically. There exist different execution paths of the algorithm $\mathcal A$, but the scheduler may choose a particular execution of $\mathcal A$, where the destinations of each robot in $H(v)$ are the same. Since there is no robot position on $V'$, the configuration symmetry cannot be deterministically broken by allowing the robots to move from $V'$. We will prove the theorem by using induction on the number of rounds.

  \begin{figure}[!b]
				\centering
			{	\includegraphics[width=0.72\columnwidth]{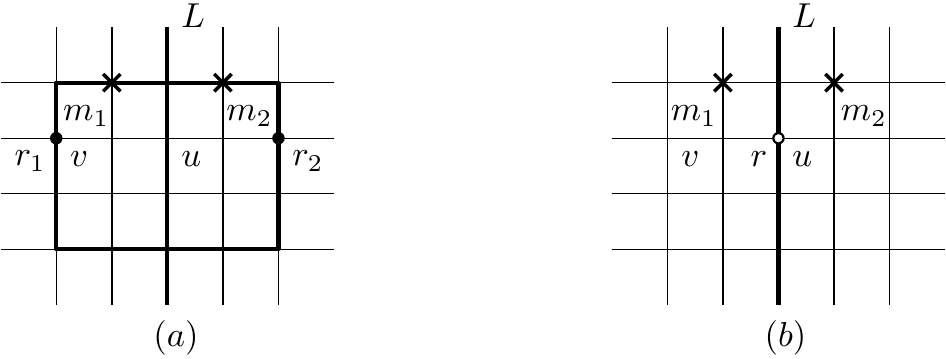}		}\vspace*{-2mm}
			\caption{(a) Each robot decides to move towards $u$ at round $t$. (b) Each robot on $H(v)$ moves towards $u$ on $L$ under the execution of $\mathcal A$ at round $t+1$. The configuration remains partitive after the movement. The circle on $u$ represents a multiplicity node.}
			\label{impoa}
		\end{figure}
	
\medskip
	  \noindent \textbf{Base Case:} By the assumption of the initial configuration, the configuration is partitive on the set $V \setminus V'$ at round $0$.

\medskip
	  \noindent \textbf{Inductive hypothesis:} Assume that the configuration is partitive on the set $V \setminus V'$ at round $t \geq1$.
	
\medskip
	  \noindent \textbf{Induction Step:} Let $r$ be an active robot at round $t$ that decides to move from node $v$ to node $u$. We need to prove that the configuration remains partitive on the set $V \setminus V'$ at round $t+1$. At round $t+1$, the following cases are to be considered.
\begin{enumerate}
    \item $v \in V\setminus V'$ and $u \in V'$. Note that at round $t$, the robots in $H(v)$ have identical views and they execute the same deterministic algorithm $\mathcal A$. As a result, there exists at least one execution of $ \mathcal A$ out of different execution paths of $\mathcal A$, where each robot in $H(v)$ moves towards the same node $u$. Each robot belonging to the other orbit $H(v')$, where $v' \neq v$, may move towards the same node $u$ by the execution of $\mathcal A$. Under this execution, the configuration remains partitive on the set $V \setminus V'$ at round $t+1$ (In Figure \ref{impoa}(a), $v$ is the node which is occupied by the robot $r_1$. Here, $V'=L$, i.e., $V'$ is the set of nodes belonging to the line of symmetry $L$. In Figure~\ref{impoa}(b), under the execution of $\mathcal A$, each robot moves towards $u$ on $L$).

    \item $v \in V'$ and $u \in V \setminus V'$. Note that, in the initial configuration, $R(0) \cap V'= \emptyset$. Therefore, there must exist some round $0<t' <t$ at which a robot $r'$ moves from a node $w \in V \setminus V'$ to the node $v \in V'$. There exist different execution paths of the algorithm $\mathcal A$, but the scheduler may choose a particular execution of $\mathcal A$, where the destinations of each robot in $H(w)$ are the same node $v$. As a consequence, the number of robots on $v$ at round $t'+1$ is $n=ak$, where $a$ denotes the number of orbits (there might be different robots moving from different orbits towards $v$). Since each robot on $V'$ lies on a multiplicity node $v$, they have identical views. As the \textit{gathering} must be ensured on a \textit{meeting node} belonging to the set $V \setminus V'$, there exists at least one execution of $\mathcal A$ in which $a$ robots from $v$ move towards $u'$ at round $t+1$, for each distinct nodes $u' \in H(u)$. Thus, the configuration remains partitive on the set $V\setminus V'$ at round $t+1$ (In Figures \ref{impob}(a) and \ref{impob}(b), under the execution of $\mathcal A$, the robots on $v$ move towards the nodes belonging to the orbit $H(u)$ and creates multiplicity on those nodes).

          \begin{figure}[!h]
          \vspace*{2mm}
				\centering
			{				\includegraphics[width=0.72\columnwidth]{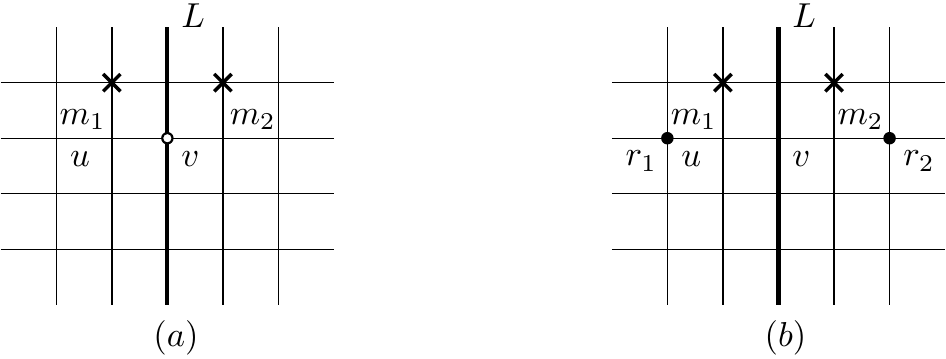}			}\vspace*{-2mm}
			\caption{(a) Each robot on $v$ decides to move towards distinct nodes of $H(u)$ at round $t$. (b) The robots on $v$ move towards distinct nodes of $H(u)$ under the execution of $\mathcal A$ at round $t+1$. The configuration remains partitive after the movement.}
			\label{impob}\vspace*{-2mm}
		\end{figure}

    \item $v \in V\setminus V'$ and $ u \in V\setminus V'$. There exists at least one execution of $\mathcal A$ in which the destinations of each robot $r'$ on the node $v'$ is some node $u'$, where $v' \in H(v)$ and $u' \in H(u)$. Since the configuration was partitive on the set $V \setminus V'$ at round $t$, the configuration remains partitive on the set $V \setminus V'$ at round $t+1$ (In Figures \ref{impoc}(a) and \ref{impoc}(b), the robots on the nodes $H(v)$ move towards the nodes belonging to $H(u)$).
		
		    \begin{figure}[!h]
\vspace*{-2mm}
				\centering
			{				\includegraphics[width=0.72\columnwidth]{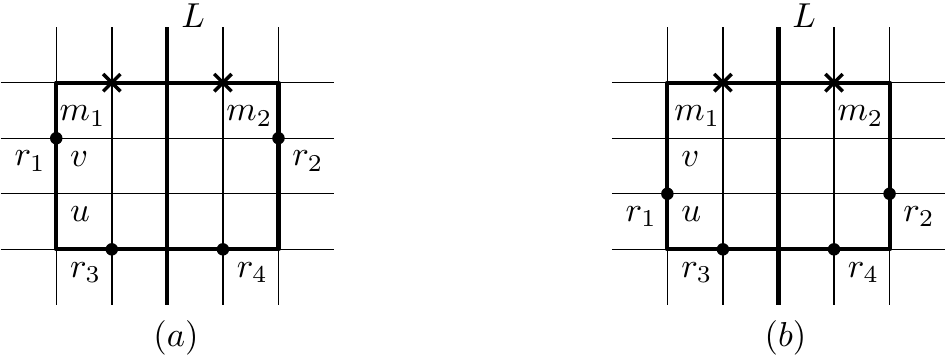}			}\vspace*{-2mm}
			\caption{(a) Each robot on $H(v)$ decides to move towards $H(u)$ on $V\setminus L$ at round $t$. (b) Each robot on $H(v)$ decides to move towards $H(u)$ on $V\setminus L$ at round $t+1$ under the execution of $\mathcal A$. The configuration remains partitive after the movement.}
			\label{impoc}\vspace*{-2mm}
			\end{figure}

		    \item $v \in V'$ and $ u \in V'$. Note that, since in the initial configuration, $R(0) \cap V'= \emptyset$, there must exist some round $0<t' <t$ at which a robot $r'$ moves from a node $w \in V \setminus V'$ to the node $v \in V'$. There exists at least one execution of $\mathcal A$, where the destinations of each robot in $H(w)$ are the same node $v$. At round $t+1$, it might be the case that each robot on $v$ moves towards $u$, under the execution of $\mathcal A$. Thus, the configuration remains partitive on the set $V\setminus V'$ at round $t+1$ (In Figures \ref{impod}(a) and \ref{impod}(b), the robots on $v$ move towards $u$).

	\begin{figure}[!h]
				\centering
			{	\includegraphics[width=0.72\columnwidth]{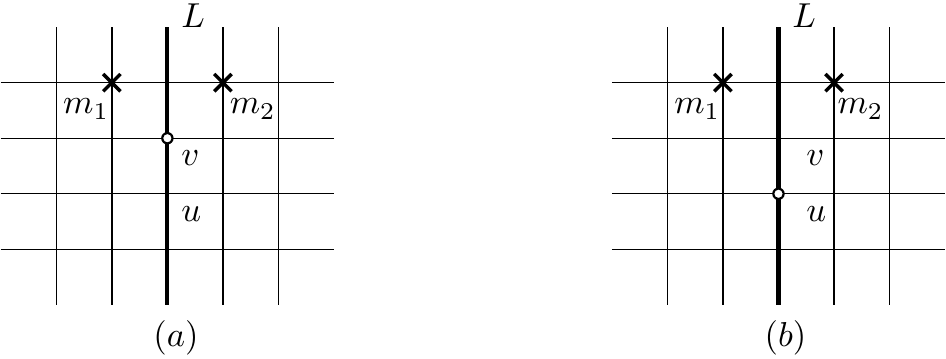}		}\vspace*{-2mm}
			\caption{(a) Each robot on $v$ decides to move towards $u$ at round $t$. (b) Each robot on $v$ moves towards $u$ at round $t+1$. The configuration remains partitive after the movement.}
			\label{impod}
			\end{figure}
		\end{enumerate}

		Starting from $C(t)$ and with the execution of $\mathcal A$, $C(t+1)$ remains partitive on the set $V\setminus V'$ at round $t+1$. Therefore, by the principle of mathematical induction, the configuration $C(t)$ remains partitive on the set $V\setminus V'$ at any round $t \geq 0$. Since the configuration remains partitive on the set $V\setminus V'$ at round $t+1$, no algorithm can ensure \textit{gathering} of the robots at a \textit{meeting node}. In fact, to ensure \textit{gathering}, there must exist a node $x \in V \setminus V'$ such that $|H(x)|=1$, but under the execution of the algorithm $\mathcal A$, the size of each orbit is $|H(x)|=k$ and $k \geq 2$, for all $x \in V \setminus V'$. This contradicts the assumption that all the robots reach a final configuration under the execution of the algorithm $\mathcal A$. Thus, \textit{gathering} cannot be ensured at a \textit{meeting node} belonging to $V \setminus V'$.
\end{proof}
		\noindent If $C(0)$ is partitive on the node set $V\setminus V'$, then from Theorem \ref{impo} it follows that there must exist at least one \textit{meeting node} $m \in V'$ where \textit{gathering} will be finalized. In this proof, we have considered the scheduler to be \textit{fully-synchronous}. Since the impossibility result holds for \textit{fully-synchronous} scheduler and the assumption of \textit{fully-synchronous} scheduler is stronger than that of \textit{asynchronous} scheduler, the impossibility result holds even for \textit{asynchronous} scheduler. Let $V'$ be the set of nodes on $L$, if $C(0) \in \mathcal I_{41}$. Otherwise, let $V'$ be the node $\lbrace c \rbrace$, if $C(0) \in \mathcal I_{42}$. Now, we have the following corollary:
	\begin{corollary} \label{i2}
		If $C(0)\in\mathcal{I}_4$, then the \textit{gathering over meeting nodes} problem is unsolvable.
	\end{corollary}
	   \begin{proof}
		First, consider the case when $C(0)\in\mathcal{I}_{41}$. This implies that $C(0)$ is partitive on the node set $V\setminus L$. According to Theorem \ref{impo}, the \textit{gathering} must be ensured at a \textit{meeting node} on $L$. Since there does not exist any \textit{meeting node} on $L$, \textit{gathering} cannot be ensured at $L$. Therefore, the \textit{gathering over meeting nodes} problem is unsolvable. The proof holds similarly in the case when $C(0)\in\mathcal{I}_{42}$, where $C(0)$ is partitive on the node set $V\setminus \lbrace c \rbrace$.
				\end{proof}
		
		\begin{corollary}
		If an initial configuration is partitive on the set $V$, then it cannot be a final configuration.
		\end{corollary}
		\begin{proof}
		Assume to the contrary that the configuration $C(0)$ is partitive on the set $V$, and $C(t)$ can be a final configuration. This implies that starting from $C(0)$, there must exist a distributed deterministic algorithm $\mathcal A$ which ensures \textit{gathering} of the robots on a \textit{meeting node}. First, consider the case when the configuration is symmetric with respect to a single line of symmetry $L$. Since the initial configuration is partitive on the set $V$, $L$ must be a line passing through the edges of the input grid graph. As the \textit{gathering} must be ensured on a \textit{meeting node} belonging to $L$, the configuration cannot be a final configuration.

Otherwise, if the configuration is symmetric with respect to rotational symmetry, the center of rotation must be a center of an edge or the center of a unit square. As the \textit{gathering} must be ensured on a \textit{meeting node} belonging to $c$, the configuration cannot be a final configuration.
		\end{proof}
In the rest of the paper, we assume that if a configuration admits a unique line of symmetry $L$, then $L$ passes through the nodes of the graph. Otherwise, if a configuration admits a rotational symmetry, then the center of rotation is a node. With this assumption, let $\mathcal U$ denote the set of all initial configurations which are ungatherable according to Corollary \ref{i2}. In other words, $\mathcal U$ represents the following collection of the configurations.
	\begin{itemize}
		\item admitting a unique line of symmetry $L$ and no \textit{meeting nodes} or robot positions on $L$.
		\item admitting rotational symmetry with no \textit{meeting node} or robot on $c$.
	\end{itemize}

		\section{Algorithm} \label{s4}

	This section describes our main algorithm $Gathering()$. The algorithm ensures \textit{gathering} over a \textit{meeting node} for all the initial configurations belonging to the set $\mathcal I \setminus \mathcal U$. The pseudo-code of the algorithm $Gathering()$ is given in Algorithm \ref{g1}. We will see later that if the \textit{meeting nodes} are asymmetric, then they can be ordered. Even if the \textit{meeting nodes} are symmetric with respect to $L$, and there exists \textit{meeting nodes} on $L$, then the \textit{meeting nodes} on $L$ are orderable.
	
			\begin{algorithm}[h]\small
		\KwIn{$C(t)=(R(t)$, $M) \in \mathcal I \setminus \mathcal U$}
		\footnotesize
		\uIf{$C(t) \in \mathcal I_{11}$ }
		{
			Each robot moves towards the \textit{meeting node} $m_s$ having the highest order with respect to $\mathcal O$ \;
		}
		\uElseIf{$C(t) \in \mathcal I_{12}$}
		{
		Each robot moves towards the \textit{meeting node} $m_z$ on $L$ having the highest order with respect to $\mathcal O'$ \;
		}
		\uElseIf{$C(t) \in \mathcal I_{13}$}
		{Each robot moves towards the \textit{meeting node} on $c$ \;
		}
		\uElseIf{$C(t) \in \mathcal I_{2}$}
		{
		$GatheringAsym()$\;
		}
		\ElseIf{$C(t) \in \mathcal I_{3}$}
		{
		$SymmetryBreaking()$ \;
		$GatheringAsym()$\;
		}
		
		\caption{Gathering()}
		\label{g1}
		\end{algorithm}

  \medskip
 First, consider the case when the \textit{meeting nodes} are asymmetric. Note that in this case, there exists a unique lexicographic largest string $\alpha_i$ (In the Figure \ref{orderingalgo}, $D$ is the unique \textit{leading corner} and $\alpha_{DA}$ is the unique largest lexicographic string). Consider an ordering $\mathcal O$ of the \textit{meeting nodes}, defined with respect to the unique \textit{leading corner}. Formally, while defining the string $\alpha_i$, let $(m_1, m_2, \ldots, m_s)$ be the ordering of \textit{meeting nodes} that appears in the \textit{string representation} of $\alpha_i$ (In Figure \ref{orderingalgo}(a), $(m_1,m_2,m_5,m_6,m_3,m_4)$ is the ordering $\mathcal O$). Similarly, if the \textit{meeting nodes} are symmetric with respect to a single line of symmetry $L$ and there exist \textit{meeting nodes} on $L$, the \textit{meeting nodes} on $L$ can be ordered according to their distances from the \textit{leading corner(s)}. Let $\mathcal O'=(m_1, m_2, \ldots m_z)$ be the ordering of the \textit{meeting nodes} on $L$, where $z$ denote the number of \textit{meeting nodes} on $L$ (In Figure \ref{orderingalgo}(b), $C$ and $D$ are the \textit{leading corners}. $(m_5, m_6)$ is the ordering $\mathcal O'$). Hence, we have the following observations.

\begin{figure}[!h]
\vspace*{-2mm}
			\centering
			{				\includegraphics[width=0.27\columnwidth]{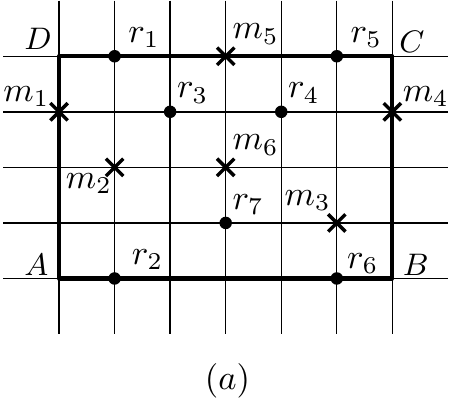}			}
			\hspace*{0.51cm}
			{				\includegraphics[width=0.27\columnwidth]{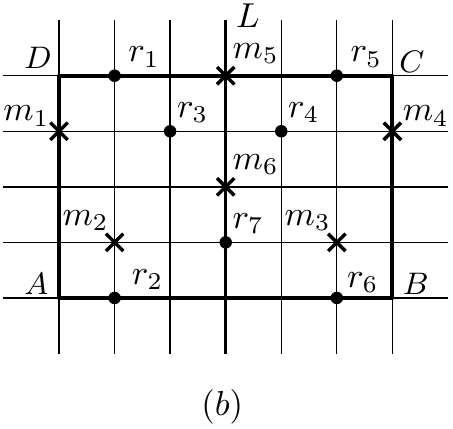}			}		\vspace*{-1mm}
			\caption{(a) The \textit{meeting nodes} are asymmetric. (b) The \textit{meeting nodes} are symmetric with respect to $L$, and there exists \textit{meeting nodes} $m_5$ and $m_6$ on $L$.}
			\label{orderingalgo}\vspace*{-4mm}
		\end{figure}

			\begin{observation} \label{o1}
		If the \textit{meeting nodes} are asymmetric, then they are orderable.
		\end{observation}
	\begin{observation} \label{o2}
	If the \textit{meeting nodes} are symmetric with respect to a unique line of symmetry $L$, and there exists at least one \textit{meeting node} on $L$, then the \textit{meeting nodes} on $L$ are orderable.
	\end{observation}

	\subsection{Gathering()}

	 In this subsection, a deterministic distributed algorithm $Gathering()$ has been proposed to solve the {\it gathering over meeting nodes problem} in an infinite grid graph. Our proposed algorithm solves the \textit{gathering} problem for all the configurations belonging to the set $\mathcal I \setminus \mathcal U$ and comprising at least two robots. The algorithm $Gathering()$ works according to the class of configurations that each robot perceives in its local configuration view. The strategy of the algorithm is to find a single \textit{meeting node} such that all the robots can agree on it and gather at that node within a finite amount of time. If $ \vert M \vert = 1$, then all the robots move towards the unique \textit{meeting node} and finalize the \textit{gathering}. So, we assume that $\vert M \vert \geq 2$. The unique \textit{meeting node} which is considered for \textit{gathering} is defined as the \textit{target meeting node}.
	
\subsubsection{$\mathcal I_1$}
	Depending on whether the initial configuration $C(0)$ belongs to $\mathcal I_{11}$, $\mathcal I_{12}$ and $\mathcal I_{13}$, the following cases are considered:
	\begin{enumerate}
\itemsep=0.8pt
	    \item $C(0) \in \mathcal I_{11}$. According to Observation \ref{o1}, since the \textit{meeting nodes} are asymmetric, they are orderable. Consider the ordering $\mathcal O$ of the \textit{meeting nodes}, defined with respect to the unique \textit{leading corner}. The \textit{meeting node} $m_s$ having the highest order with respect to $\mathcal O$ is selected as the \textit{target meeting node}. All the robots move towards $m_s$ and finalize the \textit{gathering} at it (In Figure \ref{orderingalgo}(a), $D$ is the \textit{leading corner} and $(m_1,m_2,m_5,m_6,m_3,m_4)$ is the ordering $\mathcal O$. $m_4$ is the \textit{meeting node} which has the highest order in $\mathcal O$. $m_4$ is selected as the \textit{target meeting node}).

\begin{figure}[!b]
				\centering
		{	\includegraphics[width=0.25\columnwidth]{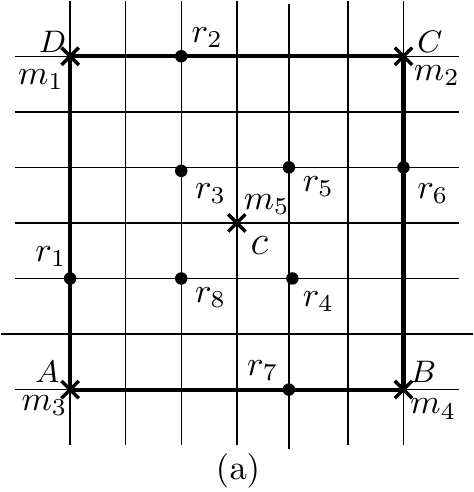}	}
		\hspace*{0.9cm}
		{	\includegraphics[width=0.23\columnwidth]{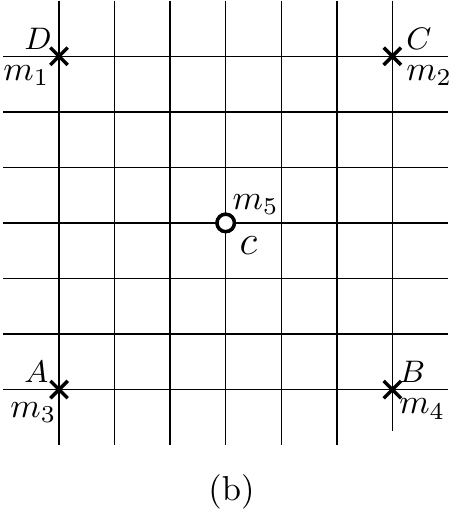}		}
			\caption{(a) $\mathcal{I}_{13}$- configuration with {\it meeting node} on $c$. $m_5$ is selected as the \textit{target meeting node}. (b) The {\it gathering} is finalized by moving the robots towards $m_5$.}
		\label{order}
	\end{figure}

	    \item $C(0) \in \mathcal I_{12}$. There exists at least one \textit{meeting node} on $L$. According to Observation \ref{o2}, the \textit{meeting nodes} on $L$ are orderable. Since the {\it meeting} nodes are fixed, the ordering remains invariant during the movement of the robots. Consider the ordering $\mathcal O'=(m_1, m_2, \ldots m_z)$ of the \textit{meeting nodes} on $L$. The \textit{meeting node} $m_z$ on $L$ having the highest order with respect to $\mathcal O'$ is selected as the \textit{target meeting node}. Each robot moves towards the \textit{meeting node} $m_z$, and the \textit{gathering} is finalized at $m_z$ (In Figure \ref{orderingalgo}(b), $C$ and $D$ are the \textit{leading corners} and $(m_5,m_6)$ is the ordering $\mathcal O'$. $m_6$ is the \textit{meeting node} on $L$ which has the highest order among all the \textit{meeting nodes} on $L$ according to $\mathcal O'$. $m_6$ is selected as the \textit{target meeting node}).
	    	
	    \item $C(0) \in \mathcal I_{13}$. There exists a \textit{meeting node} (say $m$) on $c$. Each robot moves towards $m$, and finalizes the \textit{gathering} at $m$ (In Figure \ref{order}(a), $m_5$ is selected as the \textit{target meeting node}. In Figure \ref{order}(b), each robot moves towards $m_5$).
	\end{enumerate}

	\begin{lemma} \label{lemma1}
	If $C(0)\in$ $\mathcal I_{1}$, then the \textit{target meeting node} remains invariant during the execution of the algorithm $Gathering()$ at any time $t>0$.
	\end{lemma}
	\begin{proof}
	Depending on whether the initial configuration $C(0)$ belongs to $\mathcal I_{11}$, $\mathcal I_{12}$ and $\mathcal I_{13}$, the following cases are considered.
	\begin{case}
	   $C(0)\in$ $\mathcal I_{11}$. Each robot agrees on the ordering $\mathcal O$ of the \textit{meeting nodes}. Since the ordering remains invariant during the movement of robots, the \textit{meeting node} having the highest order in $\mathcal O$ also remains invariant. As a result, the \textit{target meeting node} remains invariant.
	\end{case}
	\begin{case}
	   $C(0)\in$ $\mathcal I_{12}$. The \textit{target meeting node} $m_z$ is selected as the \textit{meeting node} on $L$ having the highest order with respect to $\mathcal O'$. Since the ordering depends on the positions of \textit{meeting nodes}, it remains invariant while the robots move towards $m_z$. Consider the case when the configuration is symmetric. Even if the configuration becomes asymmetric because of a possible pending move, $L$ remains uniquely identifiable, as it is also the line of symmetry for $M$.
	\end{case}

	\begin{case}
	   $C(0)\in$ $\mathcal I_{13}$. The \textit{meeting node} $m$ on $c$ is selected as the \textit{target meeting node}. Since $c$ is also the center of the rotational symmetry of the \textit{meeting nodes}, it remains invariant while the robot moves towards it.
	   \end{case}
	   Hence, the \textit{target meeting nodes} remains invariant during the execution of the algorithm at any time $t>0$.
		\end{proof}

	\subsubsection{$\mathcal I_2$}

		Assume that the initial configuration $C(0) \in \mathcal I_2$. In this case, the \textit{meeting nodes} are symmetric, but the configuration is asymmetric. There does not exist any \textit{meeting node} on $L \cup \lbrace c \rbrace$. Here, each robot executes $GatheringAsym()$. The overview of the procedure $GatheringAsym()$ is discussed as follows.
	
\medskip
	\noindent \textbf{Overview of the procedure:}
	 The procedure comprises the following phases: \textit{Guard Selection and Placement, Creating Multiplicity on Target Meeting Node} and \textit{Finalization of Gathering}. Since the robots are oblivious, each robot determines its current phase by analyzing the current configuration in its local configuration view. Note that, since the \textit{meeting nodes} are symmetric, any ordering of the \textit{meeting nodes} in the initial configuration depends on the robot positions and may change during the movement of the robots. In order to fix a particular ordering of the \textit{meeting nodes}, a robot denoted as a \textit{guard} is selected and placed in the \textit{Guard Selection and Placement} phase. Each \textit{non-guard} robot moves towards the \textit{target meeting node} in the \textit{Creating Multiplicity on Target Meeting Node} phase. The \textit{guard} is selected and placed in such a way that during the execution of the procedure $GatheringAsym()$, it remains uniquely identifiable by the other robots. The main strategy of the algorithm is to maintain the invariance of the \textit{target meeting node} in the \textit{Creating Multiplicity on Target Meeting Node} phase. Finally, in the \textit{Finalization of Gathering} phase, the \textit{guard} moves towards the \textit{target meeting node}. While the \textit{guard} moves, it moves in a shortest path.
	 	
\medskip
	\noindent\textbf{\emph{Guard Selection and Placement:}} In this phase, a single robot is selected as the \textit{guard}. The \textit{guard} is selected and placed in such a way that it remains uniquely identifiable by the other robots during the execution of the procedure $GatheringAsym()$. Let $MER_F$ denote the \textit{minimum enclosing rectangle} of all the \textit{meeting nodes}. First, assume that the \textit{meeting nodes} are symmetric with respect to a unique line of symmetry $L$, and there does not exist any \textit{meeting node} on $L$. Since the configuration is asymmetric, there always exists a unique \textit{key corner}. As a result, a unique robot with the maximum configuration view exists. Let $d_1$ denote the maximum distance between a \textit{meeting node} from $L$. Similarly, let $d_2$ denote the maximum distance between a robot position from $L$. Next, we consider the following cases.
		\begin{enumerate}
	 	    \item There exists at least one robot position outside the rectangle $MER_F$. This implies that there exists at least one robot position at a distance $d_2>d_1$ from $L$. If there are multiple robots at a distance $d_2$, consider the robot with the maximum configuration view. Let $r$ be the robot with the maximum configuration view. $r$ is selected as the \textit{guard}, and it moves towards an adjacent node away from $L$. This movement results in creating a unique robot which is at the maximum distance from $L$ (In Figures \ref{guardnew}(a) and \ref{guardnew}(b), $r_2$ and $r_7$ are the robots outside the $MER_F$ and at the farthest distance from $L$. $r_7$ is the robot with the maximum configuration view as $B$ is the \textit{key corner}. $r_7$ move towards an adjacent node).

\begin{figure}[!h]
			\centering
			{				\includegraphics[width=0.850\columnwidth]{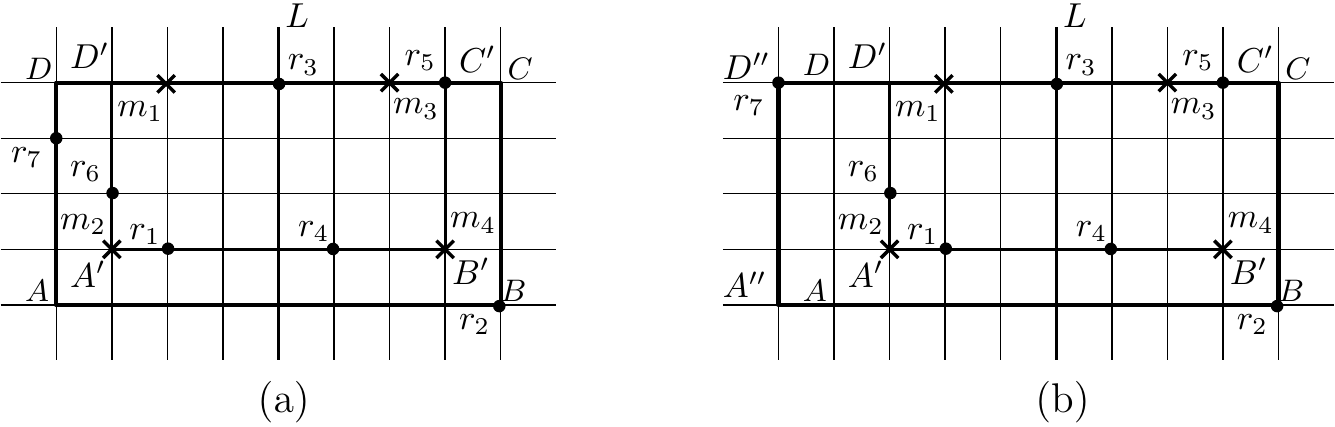}			}\vspace*{-1mm}
					\caption{ (a) $A'B'C'D'$ is the minimum enclosing rectangle $MER_F$ of all \textit{meeting nodes}. $ABCD$ is the $MER$. (b) $r_7$ is selected as a \textit{guard} and moves towards an adjacent node away from $L$. Finally, it moves towards its closest corner. The transformed $MER$ is $A''B C D''$.}			\label{guardnew}\vspace*{-1mm}
		\end{figure}

	 	    \item Each robot position is either inside or on the rectangle $MER_F$. This implies that any robot can be at the maximum distance $d_2$ from $L$ and $d_2 \leq d_1$. Consider the robot farthest from $L$. If there are multiple such robots, consider the robot $r$ with the maximum configuration view. $r$ is selected as the \textit{guard} and it moves towards an adjacent node away from $L$. $r$ continues its movement and the moment $r$ reaches a node which is outside the rectangle $MER_F$, $d_2$ becomes greater than $d_1$. The rest of the procedure follows similarly from the previous case. \end{enumerate}
	 	Once the \textit{guard} becomes the unique farthest robot from $L$, it moves towards the closest corner of $MER$ in the direction parallel to $L$. If the \textit{guard} is closest to two corners of $MER$, then it moves towards an arbitrary corner (In Figure \ref{guardnew}(b), $r_7$ moves towards $D''$). The procedure follows similarly when the \textit{meeting node} admits rotational symmetry, and there does not exist any \textit{meeting node} on $c$. In that case, $d_1$ and $d_2$ are defined as the distances measured from $c$. The pseudo-code corresponding to this phase is given in Algorithm \ref{algo2}.
	 		
\begin{algorithm}[h]\small
			\KwIn{$C(t)=(R(t)$, $M)$}
			\footnotesize
			\uIf{there exists at least one robot position outside $MER_F$}
			{
			\uIf{there exists exactly one robot position $r$ outside $MER_F$}
			{
			$r$ is selected as the \textit{guard} \;
			}
			\Else{The robot $r$ farthest from $L \cup \lbrace c \rbrace$ and with the maximum configuration view is selected as the \textit{guard} \;
			}
			
			$r$ moves towards an adjacent node away from $L \cup \lbrace c \rbrace$ and finally towards its closest corner\;
			}
			\ElseIf{each robot is inside or on $MER_F$}
			{
			\uIf{there exists a unique robot $r$ farthest from $L \cup \lbrace c \rbrace$ }
			{ $r$ is selected as the \textit{guard}\;
			}
			\Else{ The robot $r$ farthest from $L \cup \lbrace c \rbrace$ and with the maximum configuration view is selected as the \textit{guard}  \;
			}
			$r$ moves towards an adjacent node away from $L \cup \lbrace c \rbrace$ and continues its movement unless it is outside $MER_F$\;
			}
				\caption{$GuardSelection()$}
			\label{algo2}
		\end{algorithm}
	
 	\begin{lemma} \label{guardproof}
	 	During the execution of the procedure $GuardSelection()$, the \textit{guard} remains uniquely identifiable by the robots.
	 	\end{lemma}
	
 	\begin{proof}
	 	First, assume that the \textit{meeting nodes} are symmetric with respect to $L$ and there does not exist any \textit{meeting nodes} on $L$. The proof follows similarly when the configuration admits rotational symmetry, and there does not exist any \textit{meeting node} on $c$. The following cases are to be considered.
	 	\setcounter{case}{0}
	 	\begin{case} \label{case111}
	 	There exists at least one robot position outside the rectangle $MER_F$. Note that there may be multiple such robots. If there exists precisely one such robot $r$, then according to the procedure $GuardSelection()$, $r$ is selected as a \textit{guard}. Otherwise, if there are multiple such robots, the robot $r$ with the maximum configuration view is selected as the \textit{guard}. $r$ moves towards an adjacent node $v$ away from $L$. The moment it reaches $v$, it becomes the unique farthest robot from $L$ and at least at a distance $d_2+1$ from $L$. While the \textit{guard} moves towards the corner, it remains the unique farthest robot from $L$. As the \textit{guard} is selected as the unique farthest robot from $L$, it remains uniquely identifiable.
	 	\end{case}
	 	\begin{case}
	 	Each robot position is inside or on the rectangle $MER_F$. In this case, the robot position farthest from $L$ is selected as a \textit{guard}. Note that there may be multiple such robots. The procedure $GuardSelection()$ ensures that the \textit{guard} is selected as the robot $r$ with the maximum configuration view. The moment $r$ moves towards an adjacent node away from $L$, it becomes the unique farthest robot from $L$. $r$ continues its movement unless it becomes the unique robot that is outside the rectangle $MER_F$. The rest of the proof follows from Case \ref{case111}.
	 	\end{case}

\vspace*{-6mm}
	 	\end{proof}	
 According to Lemma \ref{guardproof}, the \textit{guard} is the unique robot which is farthest from $L \cup \lbrace c \rbrace$. As a result, the \textit{guard} will not have any symmetric image with respect to $L \cup \lbrace c \rbrace$, and the configuration remains asymmetric. Once the \textit{guard} is selected and placed, we consider the corner of $MER$, which is occupied by the \textit{guard}. Starting from that corner, we scan the entire $MER$ in the direction parallel to the \textit{string direction} and associate each node $v$ to $\lambda(v)$. As a result, we would get a binary string. Consider the ordering of the \textit{meeting nodes} according to their positions in the string representation. We define the particular ordering by $\mathcal O''$.
	
\begin{figure}[!h]
\vspace{2mm}
			\centering
			{				\includegraphics[width=0.94\columnwidth]{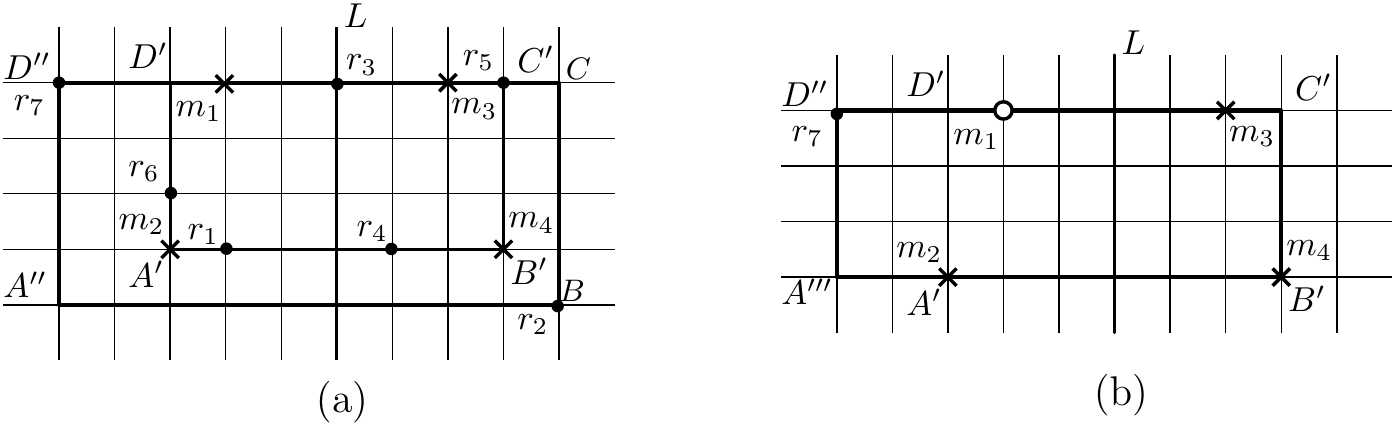}			}
					\caption{ (a) The configuration after the execution of $GuardSelection$. $m_1$ is selected as the \textit{target meeting node}. (b) Each \textit{non-guard} moves towards $m_1$ and creates a multiplicity on $m_1$. The $MER$ becomes $A'''B'C'D''$. The circle on $m_1$ represents a robot multiplicity on $m_1$.}
			\label{multinew}\vspace*{-2mm}
		\end{figure}

\medskip
	\noindent \textbf{\emph{Creating Multiplicity on Target Meeting Node:}} In this phase, each \textit{non-guards} moves towards the \textit{target meeting node} $m$. Since $R \cup M$ is asymmetric, there exists a unique ordering of the \textit{meeting nodes} with respect to the \textit{guard}. Note that the ordering remains invariant unless the \textit{guard} moves. Each robot agrees on the ordering $O''$ of the \textit{meeting nodes}. The \textit{target meeting node} $m$ is selected as the \textit{meeting node} that is closest to the \textit{guard}. If there are multiple such \textit{meeting nodes}, consider the \textit{meeting node} $m$ that has the minimum order in $\mathcal O''$ as the \textit{target meeting node}. Let $n_1$ denote the total number of distinct robot positions. If $n_1\geq 3$, then each \textit{non-guard} moves towards $m$ by executing $MakeMultiplicity()$ (In Figures \ref{multinew}(a) and \ref{multinew} (b), $m_1$ is the closest \textit{meeting node} from the \textit{guard}. Each \textit{non-guard} moves towards $m_1$). This would result in creating a robot multiplicity at $m$. Note that during this movement, the robots may create multiplicity on a \textit{meeting node} other than $m$. We have to ensure that, during this phase, $m$ remains invariant and uniquely identifiable. All the \textit{non-guard} robots move towards $m$ sequentially, i.e., the \textit{non-guard} robots that are closest to $m$ move first. The pseudo-code corresponding to this phase is given in Algorithm \ref{algo3}.
	
		\begin{algorithm}[h]\small
			\KwIn{$C(t)=(R(t)$, $M)$}
			\footnotesize
			\uIf{there exists a unique meeting node $m$ closest to the \textit{guard}}
			{
			Each robot selects $m$ as the \textit{target meeting node}\;
			}
			\Else{ Let $m$ be the closest {\it meeting node} that has the minimum order in $\mathcal O''$ \;
			Each robot selects $m$ as the {\it target meeting node}\;}
			Let $r$ be a closest \textit{non-guard} robot which is not on $m$\;
		$r$ moves towards $m$ in a shortest path \;
			\caption{$MakeMultiplicity()$}
			\label{algo3}
		\end{algorithm}
			\begin{lemma}
During the execution of $MakeMultiplicity()$, the \textit{target meeting node} remains invariant. \label{correckey}
\end{lemma}

\begin{proof}
Since the configuration is asymmetric, there exists a unique ordering $\mathcal O''$ of the \textit{meeting nodes} with respect to the unique \textit{guard}. Note that during the execution of $MakeMultiplicity()$, the \textit{guard} does not move. The following cases are to be considered.
	\setcounter{case}{0}
		\begin{case} \label{case11}
			 The \textit{meeting nodes} are symmetric with respect to $L$. If there exists a unique \textit{meeting node} $m$ closest to the \textit{guard}, it is selected as the \textit{target meeting node}. Since the position of the \textit{meeting nodes} and the \textit{guard} remains fixed during the execution of $MakeMultiplicity()$, $m$ remains invariant. Otherwise, there are multiple closest \textit{meeting nodes} from the \textit{guard}. In that case, the \textit{target meeting node} $m$ is selected as the \textit{meeting node} closest to the \textit{guard} and that has the minimum order in $\mathcal O''$. Since the \textit{guard} does not move, the ordering $\mathcal O''$ remains invariant. Hence, $m$ remains invariant during the execution of $MakeMultiplicity()$.
			\end{case}
			\begin{case}
				The \textit{meeting nodes} are symmetric with respect to rotational symmetry. The \textit{target meeting node} is selected similarly, as in Case \ref{case11}. The rest of the proof follows similarly.
			\end{case}

\vspace*{-6mm}
\end{proof}

\begin{figure}[!b]
			\centering
			{
				\includegraphics[width=0.92\columnwidth]{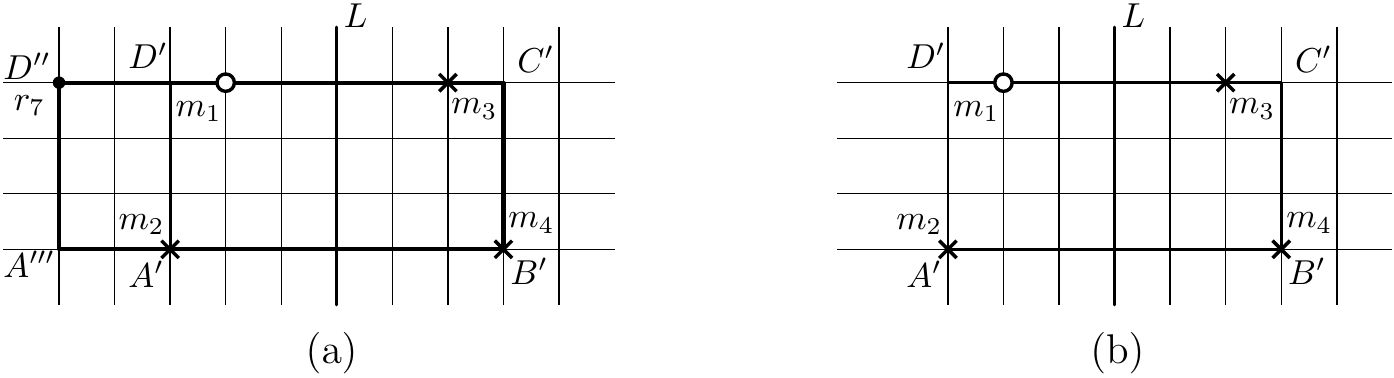}
			}\vspace*{-1mm}
		\caption{ (a) The configuration after the execution of $MakeMultiplicity()$. (b) $r_7$ moves towards $m_1$ and finalizes the \textit{gathering}. The $MER$ becomes $A'B'C'D'$.}
			\label{finalnew}
		\end{figure}
\noindent \textbf{\emph{Finalization of Gathering:}} In this phase, the \textit{guard} executes $GuardMovement()$. Note that the robots are endowed with \textit{local-weak multiplicity detection} capability. If $n_1 = 2$, then the \textit{guard} can identify that it does not lie on a robot multiplicity node and on a \textit{meeting node}. The \textit{guard} would start moving towards the other robot position in a shortest path. All the other robots on the \textit{target meeting node} $m$ would identify that they are already on a multiplicity node and they would not move. As a result, the robots on $m$ would remain on $m$. Since the \textit{target meeting node} $m$ is selected as one of the \textit{meeting nodes} closest to the \textit{guard}, while the \textit{guard} moves towards $m$ in a shortest path, it would not lie on any \textit{meeting node} other than $m$ in its movement path. Eventually, the \textit{guard} would finalize the {\it gathering} on $m$ (In Figures \ref{finalnew}(a) and \ref{finalnew}(b), $r_7$ observes that is not on a multiplicity node. It would move towards $m_1$). The pseudo-code corresponding to this phase is given in Algorithm \ref{algo4}.

	\begin{algorithm}[h]\small
			\KwIn{$C(t)=(R(t)$, $M)$}
			
			\footnotesize
					\If{$n_1=2$}
			{	Let $r$ be the robot that does not lie on a robot multiplicity node and on a {\it meeting} node\;
		
					$r$ moves towards the other robot position \;
			}
					\caption{GuardMovement()}
			\label{algo4}
			\end{algorithm}

\subsubsection{$\mathcal I_{3}$}
This subsection considers all initial configurations belonging to $\mathcal I_3$. Each robot executes $SymmetryBreaking()$. The algorithm description of the procedure $SymmetryBreaking()$ is as follows.

\begin{figure}[!h]
\vspace*{2mm}
		\centering
		{			\includegraphics[width=0.29\columnwidth]{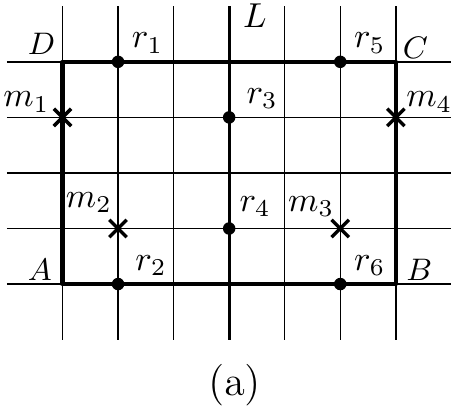}		}
		\hspace*{0.61cm}
		{			\includegraphics[width=0.29\columnwidth]{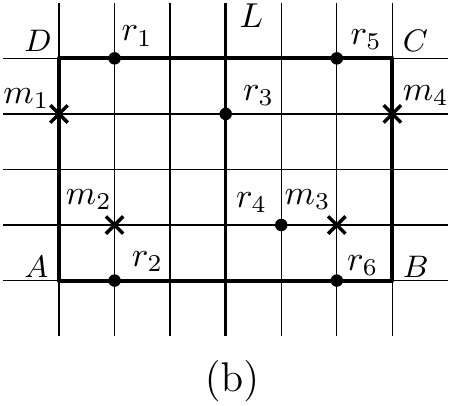}		}
		\caption{ (a) $\mathcal{I}_{31}$-configuration. (b) Transformation of $\mathcal{I}_{31}$-configuration into an asymmetric configuration by the movement of the robot $r_4$.}
		\label{alg1}
	\end{figure}

 \noindent \textbf{\emph{Symmetry Breaking}}: In this phase, all the symmetric configurations that can be transformed into asymmetric configurations are considered. A unique robot is identified that allows the transformation. We have the following cases.
 \begin{enumerate}
     \item $C(t)\in\mathcal{I}_{31}$. In this class of configurations, at least one robot exists on $L$. Let $r$ be the robot on $L$ having the maximum configuration view. $r$ moves towards an adjacent node that does not belong to $L$ (In Figures \ref{alg1}(a) and \ref{alg1}(b), $C$ and $D$ are the \textit{key corners}. As a result, $r_4$ is the unique robot on $L$ that has the maximum configuration view. $r_4$ moves towards an adjacent node away from $L$, and the configuration becomes asymmetric). Note that the configuration may not belong to $\mathcal I_{21}$ as it may contain multiplicities.
     		\begin{figure}[!h]
			\centering
			{				\includegraphics[width=0.250\columnwidth]{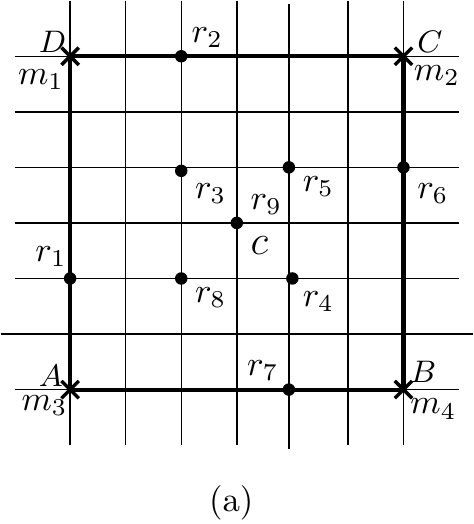}			}
			\hspace*{0.71cm}
			{				\includegraphics[width=0.250\columnwidth]{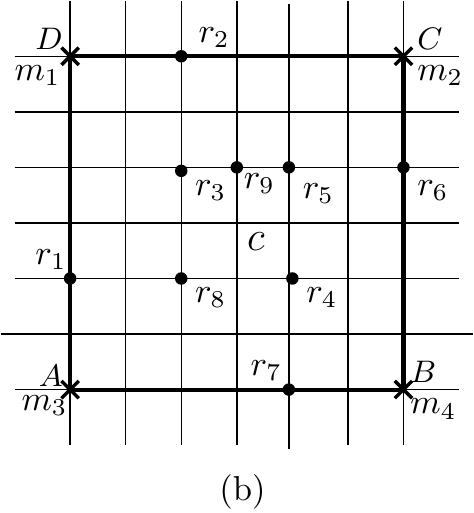}			}
			\caption{ (a) $\mathcal{I}_{32}$-configuration. (b) Transformation of $\mathcal{I}_{32}$-configuration into an asymmetric configuration by the movement of $r_9$ on $c$.}
			\label{alg11}\vspace*{-2mm}
		\end{figure}
     \item $C(t)\in\mathcal{I}_{32}$. In this class of configurations, there exists a robot (say $r$) on $c$. The robot $r$ moves towards an adjacent node (In Figures~\ref{alg11}(a) and~\ref{alg11}(b), the robot $r_9$ on $c$ moves towards an adjacent node and the configuration transforms into an asymmetric configuration). If the configuration admits rotational symmetry with multiple lines of symmetry and there is a robot $r$ at the center, $r$ moves towards an adjacent node. This movement creates a unique line of symmetry $L'$. However, the new position of $r$ might have a multiplicity. If that happens to be the robot with the maximum view on $L'$, moving robots from there might still result in a configuration with a line of symmetry. Even so, the unique line of symmetry $L'$ would still contain at least one robot position without multiplicity, and the number of robot positions on $L'$ will be strictly less than the number of robots on the line of symmetry in the original configuration. Thus, the repeated execution of the procedure $SymmetryBreaking()$ guarantees to transform the configuration into an asymmetric configuration within a finite amount of time.
 \end{enumerate}

 \noindent The pseudo-code corresponding to this phase is given in Algorithm \ref{algo5}. Once the configuration is transformed into an asymmetric configuration, the robots execute $GatheringAsym()$. Suppose a robot multiplicity is created during the execution of $SymmetryBreaking()$. In that case, the robot that is farthest from $L \cup \lbrace c \rbrace$ and does not lie on a multiplicity is selected as the \textit{guard}. If there are multiple such farthest robots, then the robot with a higher configuration view is selected as a \textit{guard}. Note that such a robot position always exists.
\begin{algorithm}[h]\small
			\KwIn{$\mathcal C(t) \in \mathcal I_3$}
			\footnotesize
			\uIf{$\mathcal C(t) \in \mathcal{I}_{31}$}
			{
			Let $r$ be the robot on $L$ with the maximum configuration view \;
			Move $r$ towards any adjacent node that does not belong to $L$\;	
			}
			
			\ElseIf{$\mathcal C(t) \in \mathcal{I}_{32}$}
			{Let $r$ be the robot on $c$ \;
				Move $r$ towards any adjacent node\;
			}
		\caption{SymmetryBreaking()}
		\label{algo5}
		\end{algorithm}
	\begin{lemma} \label{consistency}
		If $C(0) \in \mathcal I\setminus \mathcal U$, then during the execution of $Gathering()$, $C(t) \notin \mathcal U$, for any $t>0$.
	\end{lemma}

	\begin{proof}
		According to Theorem \ref{i2}, the ungatherable configurations are those configurations
		\begin{enumerate}
     \itemsep=0.85pt
		    \item admitting a unique line of symmetry $L$ and without any robot or \textit{meeting nodes} on $L$.
		    \item admitting rotational symmetry with center of rotation $c$ and without a robot or \textit{meeting node} on~$c$.
		\end{enumerate}
		 	Depending on whether the initial configuration $C(0)$ is in $\mathcal I_1$, $\mathcal I_2$, or $\mathcal I_3$, the following cases are to be considered.
		\setcounter{case}{0}
		\begin{case}
		    Consider the case when $C(0) \in \mathcal I_1$. This includes all those configurations where the \textit{meeting nodes} are either asymmetric or symmetric with at least one \textit{meeting node} on $L \cup \lbrace c \rbrace$. Since the \textit{meeting nodes} are fixed nodes located on the nodes of the grid, $C(t) \notin \mathcal U$, for any $t>0$.  \vspace*{-2mm}
		\end{case}
		\begin{case}\label{case1s}
		Consider the case when $C(0) \in \mathcal I_2$. Since the configuration is asymmetric, there exists a unique \textit{key corner}. According to Lemma \ref{guardproof}, during the execution of the procedure $GuardSelection$, the configuration remains asymmetric, as the \textit{guard} contains no symmetric image with respect to $L \cup \lbrace c \rbrace$. Note that the \textit{guard} is the unique farthest robot from $L \cup \lbrace c \rbrace$. Since the \textit{guard} does not move in the \textit{Creating Multiplicity on Target Meeting node} phase and each \textit{non-guards} moves towards $MER_F$, the configuration remains asymmetric at any time $t>0$. During the execution of $GuardMovement()$, the \textit{guard} moves towards the \textit{target meeting node} and finalizes the \textit{gathering}. Hence, $C(t') \notin \mathcal U$, for any $t'>t$, where $t'$ denotes any instant of time after the execution of $GuardMovement()$.  \vspace*{-2mm}
		\end{case}
		\begin{case} \label{case2s}
			Consider the case when $C(0) \in \mathcal I_3$. There exists at least one robot position on $L \cup \lbrace c \rbrace$. At some time $t>0$, the configuration transforms into an asymmetric configuration by the execution of $SymmetryBreaking()$. Hence, the configuration remains asymmetric at $t'>t$, similar as in Case \ref{case1s}.  \vspace*{-2mm}
		\end{case}
		Hence, if $C(0) \notin \mathcal U$, then during the execution of $Gathering()$, $C(t) \notin \mathcal U$, for any $t>0$.
				\end{proof}
		
	\noindent Without loss of generality, let $m$ be the {\it target} {\it meeting node}, selected after \textit{guards} placement at time $t$. Let $d(t)=\sum\limits_{r_i\in R(t)}d(r_i(t),m)$.
	\begin{theorem}
		If $C(0) \in \mathcal I\setminus \mathcal U$ with $n \geq 2$, then by the execution of the algorithm $Gathering()$, the gathering over meeting nodes problem is solved within finite time.
	\end{theorem}
	\begin{proof}
	According to Lemma \ref{consistency}, if the initial configuration $C(0) \notin \mathcal U$, then $C(t) \notin \mathcal U$, for any $t>0$. Depending on whether the initial configuration $C(0)$ is in $\mathcal I_1$, $\mathcal I_2$, or $\mathcal I_3$, the following cases are to be considered.  \vspace*{-2mm}
	\setcounter{case}{0}
	\begin{case}
	    $C(t) \in \mathcal I_1$. According to Lemma \ref{lemma1}, the \textit{target meeting node} $m$ remains invariant. Let $t'>t$ be an arbitrary point of time at which at least one robot starts moving towards $m$. Therefore, $d(t')=\sum\limits_{r_i\in R(t')}d(r_i(t')$, $m)$ and $d(t')< d(t)$. This implies that eventually, all the robots will reach $m$ and the \textit{gathering} is finalized at $m$ within a finite amount of time.  \vspace*{-2mm}
	    \end{case}
	    \begin{case} \label{case2}
	    $C(t) \in \mathcal I_2$. First, consider the execution of $Make Multiplicity()$. Let $t' > t$ be an arbitrary point of time after the \textit{guard selection and placement} phase. Assume that at least one \textit{non-guard} robot has completed its LCM cycle at $t'$. We have $d(t')=\sum\limits_{r_i\in R(t')}d(r_i(t')$, $m)$. According to Lemma \ref{correckey}, the \textit{target meeting node} remains invariant. If there is at most one robot position which is not on $m$ at time $t'$, then the execution of $Guard Movement()$ is started. Otherwise, let $r$ be any \textit{non-guard} robot which has computed its LCM cycle at time $t'$. Since $r$ has moved at least one node closer to $m$, we have $d(t')<d(t)$. This implies that eventually, all the \textit{non-guard} robots will reach $m$ and execution of $Guard Movement()$ will be started.
	
		\noindent Next, we consider the execution of procedure $Guard Movement()$. Assume that at time $t''$, the \textit{guard} (say $r$) starts moving towards $m$ in a shortest path. At $t''$, $d(t'')=d(r(t''),m)$ (All other robots are already on $m$). Since $r$ would move at least one node closer to $m$, $d(t_1)<d (t'')$ at $t_1>t''$. Let $t_2>t_1>t''$ be the point of time at which $r$ reaches $m$. As $d(r(t_2),m)=0$, $d(t_2)=0$. Therefore, eventually all the robots will finalize {\it gathering} at $m$ . \vspace*{-1mm}
		\end{case}
		\begin{case}
		   $C(t) \in \mathcal I_3$. In this case, $SymmetryBreaking()$ is executed. The transformed configuration becomes either asymmetric, or it may admit a unique line of symmetry. By the repeated execution of $SymmetryBreaking()$, the configuration becomes asymmetric. The rest of the proof follows from Case \ref{case2}.
		\end{case}
	Hence, execution of the algorithm $Gathering()$ eventually solves the {\it gathering} {\it over meeting nodes problem} within finite time.\vspace*{-1mm}
	\end{proof}

	\section{Lower bounds} \label{s6}

	\noindent In this section, we study the efficiency of our algorithm in terms of the total number of moves executed by the robots. Let $n$ denote the total number of robots deployed on the nodes of the input grid graph. Consider a configuration $C(t)$, where the dimension of $ { MER}$  is $1 \times (n+1)$. Assume that there is a single \textit{meeting node} which is placed at one of the corners of $MER$, and all the other $n$ nodes of the grid are occupied by the $n$ robots (In Figure \ref{ana}, $m$ is the single \textit{meeting node} and $MER=AB$). Since the \textit{gathering problem} requires all the $n$ robots to be placed at the unique \textit{meeting node}, the total number of moves executed by the robots is given by,
	$1+2+3 \ldots +n= \dfrac{n(n-1)} {2}$.
	Hence, any algorithm solving the \textit{gathering over meeting nodes problem} requires $\Omega (n^2)$ moves.
	
	\begin{figure}[!h]
				\centering
		{			\includegraphics[width=0.40\columnwidth]{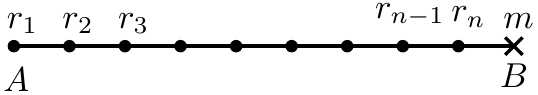}		}\vspace*{-1mm}
		\caption{A configuration showing the lower bound in terms of the number of movements required to finalize the \textit{gathering}.}
		\label{ana}
	\end{figure}

	\noindent Next, assume that $D$ = $max\lbrace p,q\rbrace$, where $p$ and $q$ are the dimensions of the initial $MER$ and if $D=\Theta(n)$, then any algorithm solves the \textit{gathering problem} in $\Omega (Dn)$ moves. Hence, we have the following theorem.
	
 \begin{theorem}
		Any gathering algorithm for solving the \textit{gathering over meeting nodes} problem requires $\Omega(Dn)$ moves.
	\end{theorem}

	\subsection{Analysis of the algorithm gathering()}
	\begin{itemize}
	    \item During the \textit{Symmetry Breaking} phase, only $O(1)$ moves are required in order to break the symmetry.
	    \item In the \textit{Guard Selection and Placement} phase, at most, one robot may be required to move away from the initial $MER$. The total number of moves during this phase is $O(D)$.
	    \item In the \textit{Creating Multiplicity on Target Meeting Node} phase, all the \textit{non-guard} robots move towards the \textit{target meeting node} in a shortest path. The total number of moves in this phase is $O(Dn)$.
	    \item Finally, in the \textit{Finalization of Gathering} phase, only the \textit{guard} moves towards the \textit{target meeting node}. The number of moves in this phase is $O(D)$.
	\end{itemize}
	Hence, we have the following result.

	\begin{theorem}
	Algorithm Gathering() solves the \textit{gathering over meeting nodes} problem in $O(Dn)$ moves if the initial configuration belongs to the set $\mathcal I\setminus U$.
	\end{theorem}

	\section{Conclusion} \label{s7}

	In this work, we have studied the \textit{gathering over meeting nodes problem} in an infinite grid where the robots have \textit{local-weak multiplicity detection} capability. A subset of all the initial configurations has been proved to be ungatherable. A deterministic distributed algorithm for solving the \textit{gathering} problem has been proposed for the remaining configurations with $n\geq 2$, where $n$ is the number of robots in the system. We have discussed the efficiency of the proposed algorithm in terms of the total number of moves executed by the robots. We have proved that the algorithm solves the \textit{gathering over meeting nodes} problem in $O(Dn)$ moves, where $D$ is the length of the longer side of the initial $MER$.
	
\medskip
	\noindent {\it Future Works:} One immediate future direction of work would be to consider the optimal algorithms for \textit{gathering}. The optimal criterion is to minimize the maximum distance traveled by any robot. Another direction of future interest would be to consider a randomized algorithm for breaking the symmetry when there are no robots or \textit{meeting nodes} on the line of symmetry and the center of rotation. It would be interesting to consider \textit{gathering} algorithms without any \textit{multiplicity detection capability}. We may also consider the \textit{gathering} problem for the configurations, where the initial configuration may contain multiplicities. If the initial deployment has multiple robots in some nodes, the problem requires a different solution.
	
\vfil\eject

	\end{document}